\documentclass[a4paper, amsfonts, amssymb, amsmath, reprint, showkeys, nofootinbib, twoside,notitlepage,onecolumn]{revtex4-1}

\def\apj{{ApJ}}                 

\def\prd{{Phys.~Rev.~D}}        
\usepackage{amsmath,amstext}
\usepackage[T1]{fontenc}
\usepackage{amssymb}
\usepackage{graphicx}
\usepackage{ae,aecompl}

\DeclareFontFamily{OT1}{pzc}{}
\DeclareFontShape{OT1}{pzc}{m}{it}{<-> s * [1.10] pzcmi7t}{}
\DeclareMathAlphabet{\mathpzc}{OT1}{pzc}{m}{it}

\usepackage{hyperref}
\usepackage{amsmath}
\usepackage{amssymb}
\usepackage{mathtools}
\usepackage{bm}
\usepackage{cleveref}
\usepackage{tensor}
\usepackage{braket}
\usepackage{enumitem}
\usepackage{mhchem}
\usepackage{amsthm}
\usepackage{nccmath}
\usepackage{mathrsfs}
\usepackage{color}

\newcommand{\E}{{\mathbb{E}}}

\newcommand{\s}{{\sigma \!\! \! \sigma}}

\def\be{\begin{equation}}
\def\ee{\end{equation}}
\def\beq{\begin{eqnarray}}
\def\eeq{\end{eqnarray}}

\theoremstyle{definition}

\theoremstyle{theorem}
\newtheorem{theorem}{Theorem}

\theoremstyle{corollary}

\begin{document}
\title{The Lorentzian geometry of relaxation}
\author{L.~Gavassino}
\affiliation{Department of Applied Mathematics and Theoretical Physics, University of Cambridge, Wilberforce Road, Cambridge CB3 0WA, United Kingdom}

\begin{abstract}
We show that relativistic theories with purely relaxational excitation spectra, such as kinetic theory and transient hydrodynamics, naturally endow the dispersion plane $\{i\omega,ik\}$ with a Lorentzian geometric structure analogous to that of the Minkowski plane $\{t,x\}$. In this picture, timelike future-directed, timelike past-directed, and spacelike directions correspond respectively to relaxation-like, unstable-like, and evanescent-like modes. Under mild structural assumptions on the underlying theory, causality constrains dispersion relations to follow spacelike trajectories on the plane. This geometric viewpoint recasts longstanding problems in relativistic matter physics as elementary geometric ones that can often be solved graphically. As applications, we derive universal constraints on dispersion relations, deviations from time dilation, the observer dependence of spectral hierarchies, the regime of validity of hydrodynamics in boosted frame, the maximal allowed diffusivity and viscosity of relativistic media, and the presence of non-hydrodynamic branch cuts in kinetic theory.
\end{abstract} 
\maketitle

\section{Introduction}

Recent years have seen renewed interest in the physics of relativistic matter, driven by applications ranging from heavy-ion collisions to astrophysics and cosmology \cite{Heinz:2013th,RomatschkeReview:2017ejr,FlorkowskiReview2018,rezzolla_book,Shibata:2019wef,Arnold:2000dr,Hannestad:2006zg}. Among the questions attracting particular attention is how special relativity constrains the linear-response properties of a medium \cite{Pu2010,GavassinoCausality2021,GavassinoSuperlum2021,HellerBounds2022ejw,GavassinoBounds2023myj,HellerHydrohedron2023jtd,Brants:2024wrx,HuiNicolis:2025aja,GavassinoRadiativeBounds:2025bxx,HoultAnisotropic:2025htt,BhattacharyyaBoostedFrames:2025hjs,Bajec:2025dqm,GavassinoDisturbing:2026klp,BrantsUVCompletions:2026qlw}. In particular, how does causality shape the excitation spectrum? Or, how do moving media generally respond to external perturbations? Despite their apparent simplicity, these questions have proven remarkably subtle \cite{Sommerfeld1914,Brillouin1914,Brillouin1960,BludmanRuderman1968,Fox1969,AharonovKomarSussking1969,FoxKuperLipson1970,KrotscheckKundt1978,Hiscock_Insatibility_first_order,Kost2000,Adams:2006sv,GavassinoDisperisons2023mad,GavassinoDiffusionCompatible:2026tvy,GavassinoHowAcausalEquaitons:2026fil}, and a unified picture is still lacking.

On the causality side, much is known about the structure of correlators from the requirement of subluminal propagation alone \cite{Bogolyubov:1990kw, LowdonStructure,Lowdon:2017gpp}. However, the resulting (currently known) constraints on dispersion relations and transport coefficients \cite{HellerBounds2022ejw,HellerHydrohedron2023jtd} are comparatively weak, as they are expressed in terms of a ``validity radius'' of the dispersion relation whose value is generally not known \emph{a priori}, and may even be infinitesimal in principle \cite{GavassinoHowAcausalEquaitons:2026fil}. The generic properties of excitation spectra of moving media are even less understood. Naively, one might expect boosted spectra to resemble their rest frame counterparts, up to time-dilation factors. In reality, this intuition applies only to patterns that comove with the bulk of the medium. In the presence of transport, relativity of simultaneity can reorder events within the perturbation, qualitatively altering the observed dynamics \cite{GavassinoDisturbing:2026klp}. A single relaxation time can split into multiple (indeed, into a continuum of) relaxation times under a Lorentz boost. A mode that is longest-lived in one reference frame may become shortest-lived in another frame. In extreme cases, a system may even appear stable in one frame and unstable in another \cite{Hiscock_Insatibility_first_order,BuchelInstabilityGaussBonnet:2026rep} (a phenomenon that occurs only when causality is violated \cite{GavassinoSuperlum2021}, but whose interpretation has puzzled the community for decades).

Here, we introduce a geometric framework for understanding the spectral properties of relativistic matter. Focusing on linear excitation modes of the form $\propto e^{ikx-i\omega t}$ that are non-oscillatory (namely, with both $k$ and $\omega$ imaginary), we show that linear-response questions can be reformulated as geometric problems involving curves, or more general sets, in the plane $\{i\omega,ik\}$. We further show that this plane is naturally endowed with a causal structure mirroring that of the Minkowski plane $\{t,x\}$. Within the resulting framework, the constraints imposed by causality and stability become geometrically transparent and can be studied systematically.

To illustrate the power of this framework, we use it to derive general quantitative answers to the following questions:
\begin{itemize}
\item[\textbf{(i)}] Suppose that we are given an exact or approximate prediction for the dispersion relation $i\omega=f(ik)$ of an isolated mode, but do not know for which values of $ik$ this prediction ceases to be reliable. Can causality alone be used to constrain its domain of validity?

\item[\textbf{(ii)}] What is the maximal deviation from time dilation that an isolated mode can undergo under Lorentz boosts?

\item[\textbf{(iii)}] Consider two modes, or two distinct sets of modes, one of which relaxes more slowly than the other. What is the smallest boost velocity at which this hierarchy can break down?

\item[\textbf{(iv)}] How does causality constrain the regime of validity of hydrodynamics in boosted frames?
\end{itemize}
The answers we will provide apply universally to relativistic theories with purely relaxational spectra, including all consistent relativistic kinetic theories, theories of transient hydrodynamics, and causal viscoelastic media.

Throughout this work, we adopt the metric signature $(-,+,...,+)$ and use natural units, $c{=}\hbar{=}k_B{=}1$. While our analysis applies to theories in $D{+}1$ dimensions, we restrict attention to modes and boosts directed along the $x$ axis.

\section{Some background: The long quest for relativistic geometry in matter}

The idea that the geometric principles of relativity should have a counterpart in theories of matter has a long history. In continuum mechanics, this connection emerges naturally from causality: the requirement that signals propagate within finite causal cones usually leads to hyperbolic equations of motion \cite{Israel_Stewart_1979,Jou_Extended,Muller_book}, whose characteristic surfaces \cite{rauch2012partial} endow material media with an effective Lorentzian geometry analogous to the geometry of spacetime itself \cite{Christodoulou2007,DisconziAcoustic_2022}.

To illustrate this idea, consider a scalar field $\phi$ obeying
$\mathcal{N}^\mu \nabla_\mu \phi
+
\mathcal{M}^{\mu\nu}\nabla_\mu\nabla_\nu\phi
=0$,
where $\mathcal{N}^\mu$ and $\mathcal{M}^{\mu\nu}=\mathcal{M}^{\nu\mu}$ are background tensors. Hyperbolicity requires that the quadratic form $\mathcal{M}^{\mu\nu}$ have Lorentzian signature $(-,+,+,+)$. The inverse matrix $(\mathcal{M}^{-1})_{\mu\nu}$ may then be interpreted as an effective Lorentzian metric, whose null cone bounds the propagation of disturbances. The resulting ``acoustic geometry'' endows the theory with a causal structure whose properties can be investigated using many of the standard techniques developed for general relativity \cite{Hawking1973,Babichev:2007dw}.

This viewpoint has proven very successful for the study of stability, well-posedness, and signal propagation in fluids \cite{DisconziAcoustic:2023rtt,GavassinoSuperlum2021}. In some contexts, such as kinetic theory and ideal-fluid dynamics, the resulting geometry can also be directly related to physics \cite{DudyinskiCausality,Christodoulou2007}. In effective field theories, however, the situation is more subtle, as gradient corrections are organized perturbatively, while the characteristics are determined by the highest derivatives retained in the equations, leading to truncation-dependent geometries. For example, in modern formulations of relativistic dissipative hydrodynamics such as Bemfica-Disconzi-Noronha-Kovtun (BDNK) theory \cite{Bemfica2019_conformal1,BemficaDNDefinitivo2020}, the characteristic geometry depends on the choice of hydrodynamic frame \cite{Kovtun2019,GavassinoInitialEFT}, namely on how the effective variables are defined away from equilibrium, and is therefore not uniquely determined by the properties of the underlying medium.

A second arena in which one might hope to uncover geometric traces of relativity is the linear excitation spectrum. Given a homogeneous equilibrium state, one may consider perturbations $\propto e^{ikx-i\omega t}$, and study the set of pairs $(\omega,k)$ for which such plane waves solve the linearized equations of motion. Historically, most attention has been focused on waves with real wavenumber $k$, corresponding to Fourier modes (the building blocks of square-integrable solutions).
Within this framework, a longstanding goal was to derive relativistic constraints on individual dispersion relations $\omega(k)$. The most famous example is the proposed inequality
\begin{equation}
\left|
\frac{d\mathfrak{Re}(\omega)}{dk}
\right|
\leq 1
\qquad\qquad
(\text{standard group-velocity bound}),
\end{equation}
which was once expected to provide a natural spectral counterpart of the lightcone. This expectation was later shown to be incorrect \cite{Sommerfeld1914,Brillouin1914,AharonovKomarSussking1969}. Even the simplest causal dissipative model, namely Cattaneo's theory $\partial_t\phi=(\partial_x^2-\partial_t^2)\phi$ \cite{cattaneo1958}, possesses wavenumbers where the group velocity diverges \cite{Pu2010}. Subsequent attempts to rescue the idea by constraining only the asymptotic behavior of dispersion relations \cite{KrotscheckKundt1978} were likewise found to fail in general. More recently, it was shown that any function $\omega(k)$ with real $k$ and $\mathfrak{Im}\,\omega\leq0$ can arise as an exact excitation branch of some causal kinetic theory \cite{GavassinoHowAcausalEquaitons:2026fil}. Hence, no relativistic geometry appears to be encoded in individual real-$k$ dispersion relations.

The main limitation of focusing on Fourier modes is that, in dissipative systems, equilibration requires $\omega$ to be complex, so if we force $k$ to remain real, we end up treating $k$ and $\omega$ (and therefore space and time) on unequal footing. This asymmetry is reflected in the fact that Lorentz boosts do not map Fourier spectra into Fourier spectra, since $k'=\gamma(k-v\omega)$, meaning that a mode with real $k$ and complex $\omega$ is generally mapped to one with complex $k'$ \cite{GavassinoSuperlum2021}.
Allowing $k$ to be complex completely removes this asymmetry issue, and has already proven much more fruitful (despite being a comparatively recent idea), as we detail below.

Consider the decomposition $e^{ikx-i\omega t}
=
e^{i(x\mathfrak{Re}k-t\mathfrak{Re}\omega)}
e^{-(x\mathfrak{Im}k-t\mathfrak{Im}\omega)}$.
The first factor is an oscillatory term, and there is no obvious causal bound on how fast a variable can locally oscillate. By contrast, the second factor controls the magnitude of the perturbation, which is directly constrained by the combination of causality and stability. In fact,
suppose that $\mathfrak{Im}k>0$. Then, the factor $e^{-x\mathfrak{Im}k}$ describes a tail that decays exponentially as $x\rightarrow+\infty$. Consider now the null worldline $x=t$. A probe traveling along this worldline cannot be influenced by the behavior of the initial data at $x<0$, being causally disconnected from it. In particular, the probe cannot ``know'' that the perturbation diverges exponentially as $x\rightarrow-\infty$, since the perturbation could as well be a localized wavepacket, based on the probe's limited knowledge. Therefore, if the theory is causal and stable, the amplitude measured along $x=t$ cannot itself diverge exponentially at late times. Evaluating the magnitude factor along the null ray gives $e^{-t(\mathfrak{Im}k-\mathfrak{Im}\omega)}$, which remains bounded only if $\mathfrak{Im}\omega\leq\mathfrak{Im}k$. Repeating the argument for $\mathfrak{Im}k<0$ yields the inequality
\begin{equation}
\mathfrak{Im}\omega\leq |\mathfrak{Im}k|  
\qquad\qquad
(\text{covariant stability bound}).
\label{covstabhydrobound}
\end{equation}
This bound is particularly useful in relativity, because it is Lorentz invariant.
Indeed, it lies at the heart of the modern stability-causality theorem \cite{GavassinoSuperlum2021} and of the first universal rigorous bounds on relativistic transport coefficients \cite{HellerBounds2022ejw,GavassinoBounds2023myj,HellerHydrohedron2023jtd}.

In this article, we show that \eqref{covstabhydrobound} is the first instance of a more general geometric structure, which becomes apparent if one focuses specifically on situations where $k$ and $\omega$ are both imaginary. In this setup, the real pairs $(i\omega,ik)$ live on a plane in which previously unrecognized patterns become apparent, and many novel results can be derived and organized in a rigorous and systematic fashion.

\newpage
\section{Geometry of the relaxation plane}
\vspace{-0.3cm}

Our starting point is the observation that the pair $(i\omega,ik)$, assumed throughout this article to be real, transforms as a Lorentz vector. Indeed, boosting the exponential factor $e^{ikx-i\omega t}$ yields a new spacetime dependence of the form $e^{ik'x'-i\omega't'}$, with
\begin{equation}\label{boost}
\begin{cases}
t=\gamma(t'+vx') \, , \\
x=\gamma(x'+vt') \, ,
\end{cases}
\qquad \Longrightarrow \qquad
\begin{cases}
i\omega=\gamma(i\omega'+vik') \, , \\
ik=\gamma(ik'+vi\omega') \, ,
\end{cases} \qquad\qquad \left(\text{with } \gamma=\dfrac{1}{\sqrt{1-v^2}}\right) \, .
\end{equation}
The dispersion plane $\{i\omega,ik\}$ therefore inherits a Lorentzian structure analogous to that of Minkowski spacetime. In particular, wavevectors $(i\omega,ik)$ can be classified according to the sign of the invariant norm $(i\omega)^2{-}(ik)^2$.

If a wavevector is timelike, then one can boost to a frame in which the corresponding exponential factor takes the homogeneous form $e^{-i\omega' t'}$. If $i\omega'>0$, namely if the wavevector is future-directed, the mode describes homogeneous relaxation in that frame, and we classify it as \textit{relaxation-like}. If instead $i\omega'<0$, corresponding to a past-directed wavevector, the profile grows exponentially in time. The medium therefore appears unstable in that frame \cite{Hiscock_Insatibility_first_order}, and we classify the mode as \textit{unstable-like}.
If, on the other hand, the wavevector is spacelike, one can boost to a frame in which the exponential factor takes the stationary form $e^{ik'x'}$. This corresponds to a planar profile that decays exponentially in space. Such profiles arise, for example, in skin-depth problems and in stationary flows past obstacles \cite{NovakWithersObstacles:2018pnv}, and we therefore classify these modes as \textit{evanescent-like}. This classification is illustrated in Fig. \ref{fig:DiffAndKot} (left panel).

To develop some intuition for this construction, let us consider a few examples. Take the diffusion equation $\partial_t \phi\,{=}\,\partial_x^2\phi$,
whose dispersion relation is $i\omega=-(ik)^2$. The invariant norm of the corresponding wavevector is $(i\omega)^2{-}(ik)^2=(ik)^2\big[(ik)^2{-}1\big]$.
At small $ik$, the norm is negative, and the mode is therefore evanescent-like. However, for $|ik|>1$, the norm becomes positive. Since $i\omega<0$, the corresponding modes are unstable-like. Hence, we recover the well-known fact that diffusion becomes unstable in some boosted frame, without explicitly performing any boost transformation.
The geometric picture immediately reveals additional information. By plotting the curve $i\omega=-(ik)^2$ in the dispersion plane (Fig. \ref{fig:DiffAndKot}, right panel), one sees that every boosted observer detects an unstable mode. Moreover, the geometric construction makes it evident that the growth rate of this unstable mode increases as the boost velocity \textit{decreases}, diverging as one approaches the rest frame.

As a second example, consider the linearized Korteweg--de Vries equation $\partial_t\phi=c_s\partial_x\phi-\partial_x^3\phi$ (with $c_s^2<1$),
whose dispersion relation is $i\omega=-c_s(ik)+(ik)^3$.
This equation provides a simple model for dispersive wave propagation in shallow water \cite{Korteweg1895}. At small wavenumbers, the linear term dominates, and the modes are evanescent-like. At sufficiently large $|ik|$, however, the cubic term causes the vector $(i\omega,ik)$ to become timelike. For large positive $ik$, one has $i\omega>0$, and the corresponding modes are therefore relaxation-like. By contrast, for large negative $ik$, one finds $i\omega<0$, so the modes become unstable-like. Observers moving in the positive-$x$ direction at any non-vanishing speed therefore detect an instability (Fig. \ref{fig:DiffAndKot}, right panel).

\begin{figure}[b!]
    \centering
\includegraphics[width=0.36\linewidth]{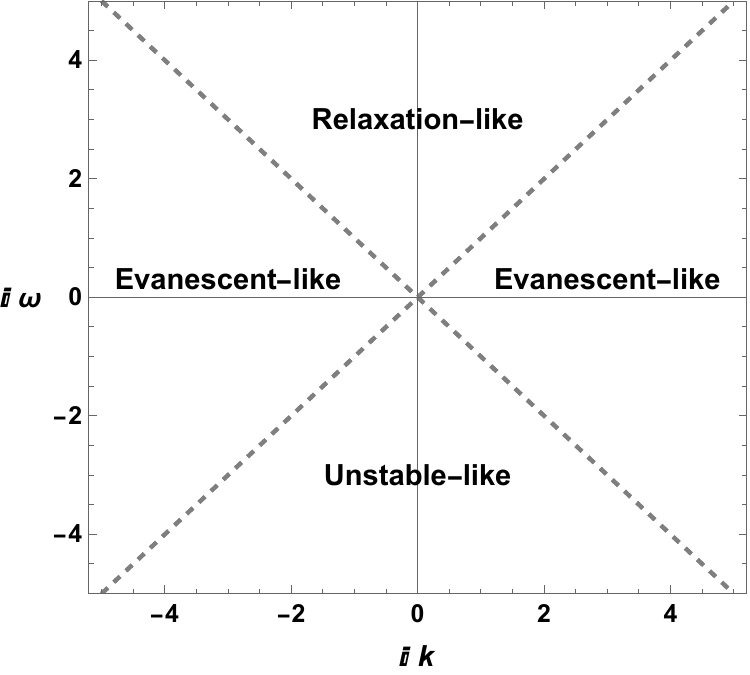}\hspace{0.08\linewidth}
\includegraphics[width=0.36\linewidth]{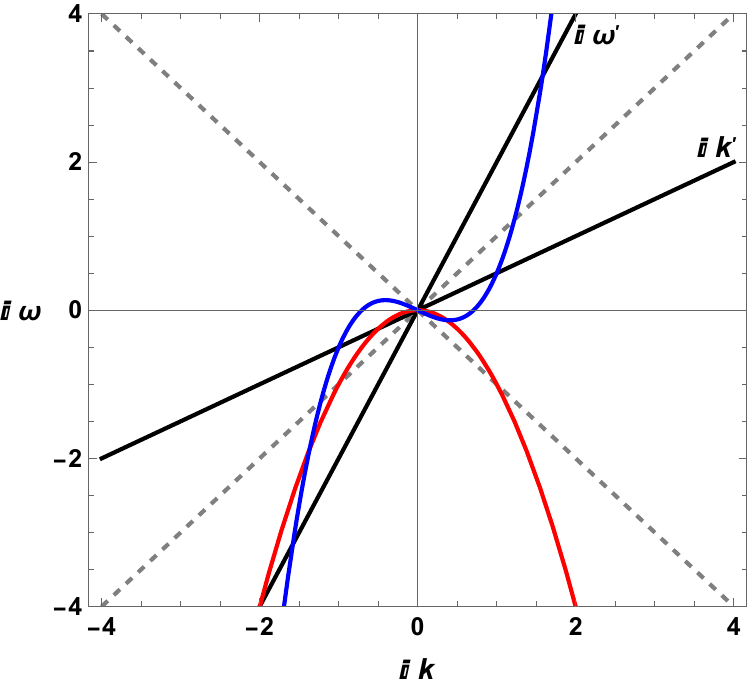}
\caption{\textit{Left panel}: Geometric classification of modes. The bound \eqref{covstabhydrobound} implies that there are no unstable-like modes in a causal and stable system.
\textit{Right panel}: Dispersion relations of the diffusion equation (red) and of the Korteweg--de Vries equation with $c_s=\frac{1}{2}$ (blue) represented in the $\{i\omega,ik\}$ plane. The intersections between these curves and the $i\omega'$ axis of a boosted observer correspond to spatially homogeneous modes that either decay (relaxation-like) or grow (unstable-like) in that observer's frame. In both examples, right-moving observers detect a growing mode, whose growth rate diverges as $v\to0$ (since the intersection happens at infinity). Intersections with the $ik'$ axis correspond to stationary evanescent modes, which describe spatially decaying, time independent profiles such as wakes and bow waves near obstacles comoving with the observer.}
    \label{fig:DiffAndKot}
\end{figure}

\newpage

Let us now consider two examples of causal theories whose spectra contain no unstable-like modes.

As a first example, consider Cattaneo's \cite{cattaneo1958} causal theory of diffusion, $\partial_t \phi=(\partial^2_x-\partial^2_t)\phi$ (see Fig.~\ref{fig:CattAndIS}, left panel), whose dispersion relation satisfies $(i\omega)^2-(ik)^2=i\omega$. Whenever the invariant norm is positive, one necessarily has $i\omega>0$. Timelike modes are therefore always future-directed, and hence relaxation-like. Conversely, all modes with $i\omega<0$ are necessarily evanescent-like. The theory thus contains no unstable modes.

As a second example, we consider Israel--Stewart theory \cite{Israel_Stewart_1979} for shear-viscous sound waves at zero chemical potential, which expressed in terms of a single field reads $(\partial^2_t{-}c_s^2 \partial^2_x)\phi+\partial_t(\partial^2_t{-}\partial^2_x)\phi=0$ (see Fig.~\ref{fig:CattAndIS}, right panel) \cite{Pu2010}. Its dispersion relation obeys $(i\omega)^2-(ik)^2=(1-c_s^2)(ik)^2/(i\omega-1)$. Thus,  positivity of the invariant norm implies $i\omega>1$, so that timelike modes are again necessarily future-directed and therefore relaxation-like. Modes with $i\omega<1$ are instead evanescent-like. Once more, no unstable-like modes are present. 

\begin{figure}[h!]
    \centering
\includegraphics[width=0.36\linewidth]{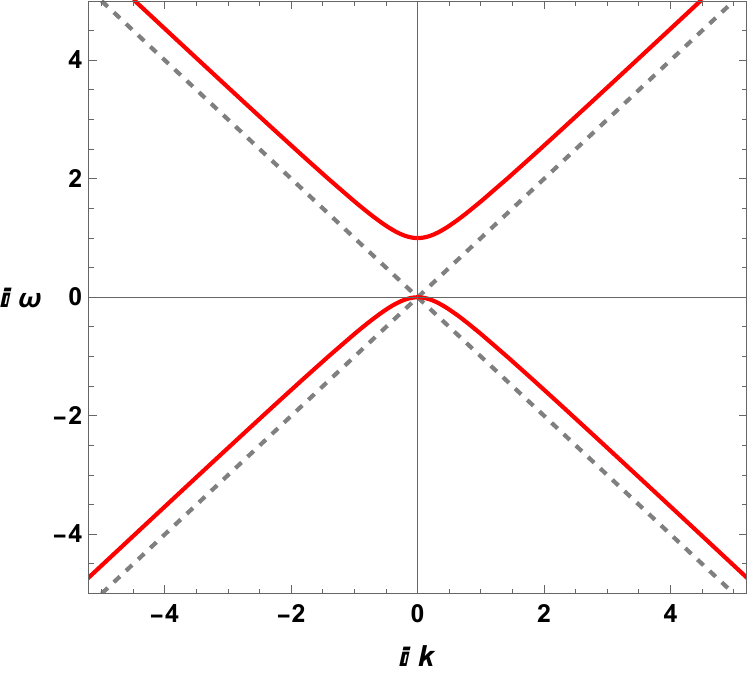}\hspace{0.08\linewidth}
\includegraphics[width=0.36\linewidth]{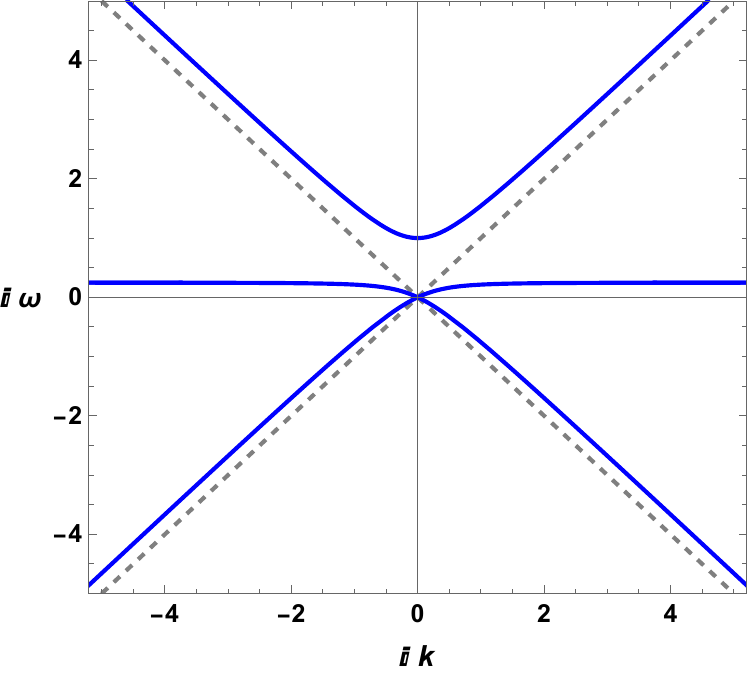}
\caption{Dispersion relations of two causal and covariantly stable theories: Cattaneo diffusion,
$\partial_t \phi+(\partial_t^2-\partial_x^2)\phi=0$
(left panel), and Israel--Stewart viscous sound,
$(\partial_t^2-c_s^2\partial_x^2)\phi+\partial_t(\partial_t^2-\partial_x^2)\phi=0$
(right panel), designed so that signals propagate at the speed of light. In both cases, the hydrodynamic modes are evanescent-like,
while the non-hydrodynamic modes are relaxation-like.}
    \label{fig:CattAndIS}
\end{figure}

\subsection{Milne and Rindler coordinates}

A natural way to extract direct physical meaning from the relaxation plane is to represent its wedges in hyperbolic coordinates. In particular, the relaxation-like sector is naturally parametrized by Milne coordinates \cite{RomatschkeReview:2017ejr}:
\begin{equation}
\begin{cases}
s =\sqrt{(i\omega)^2-(ik)^2} \, , \\
v =ik/i\omega \, ,
\end{cases}
\qquad \Longleftrightarrow \qquad
\begin{cases}
i\omega=s \gamma \, , \\
ik = s \gamma v \, .
\end{cases}
\end{equation}
\begin{figure}[h!]
    \centering
\includegraphics[width=0.37\linewidth]{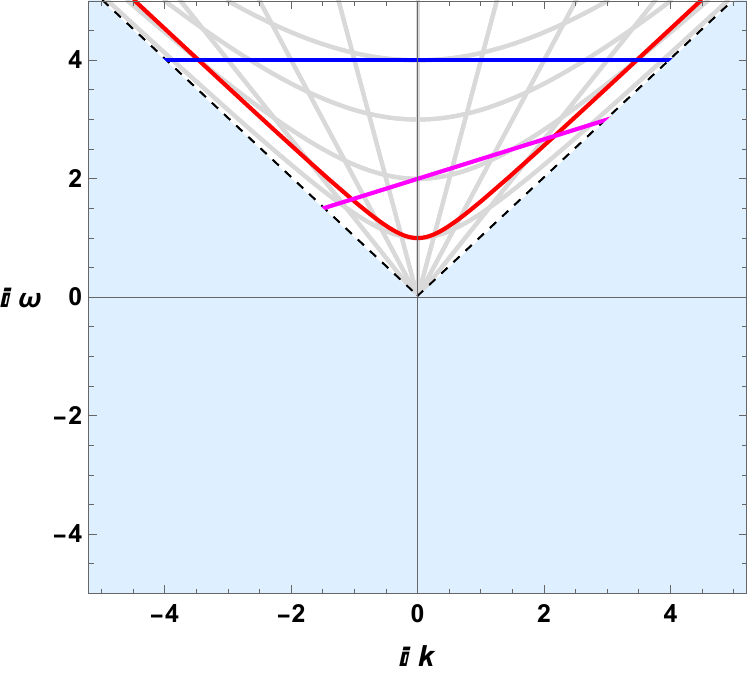}\hspace{0.08\linewidth}
\includegraphics[width=0.35\linewidth]{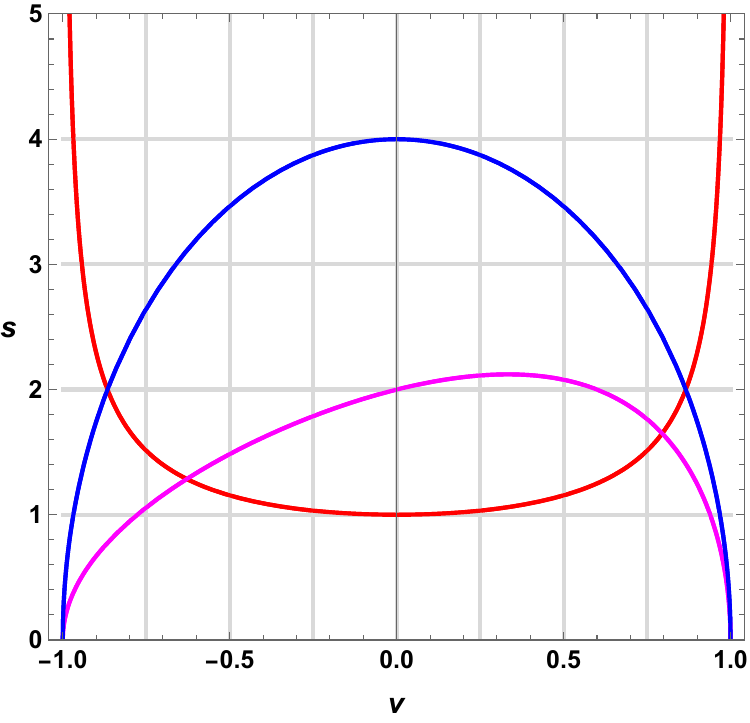}
\includegraphics[width=0.37\linewidth]{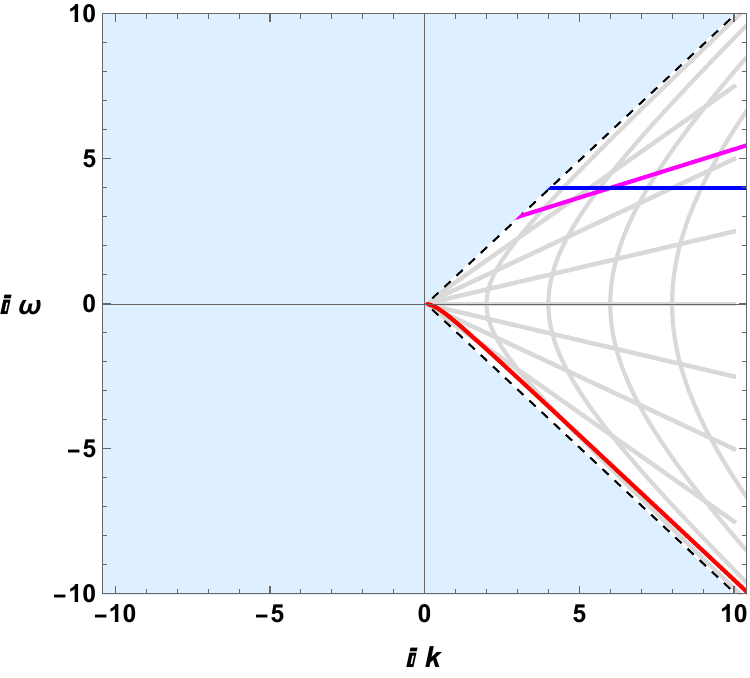}\hspace{0.08\linewidth}
\includegraphics[width=0.35\linewidth]{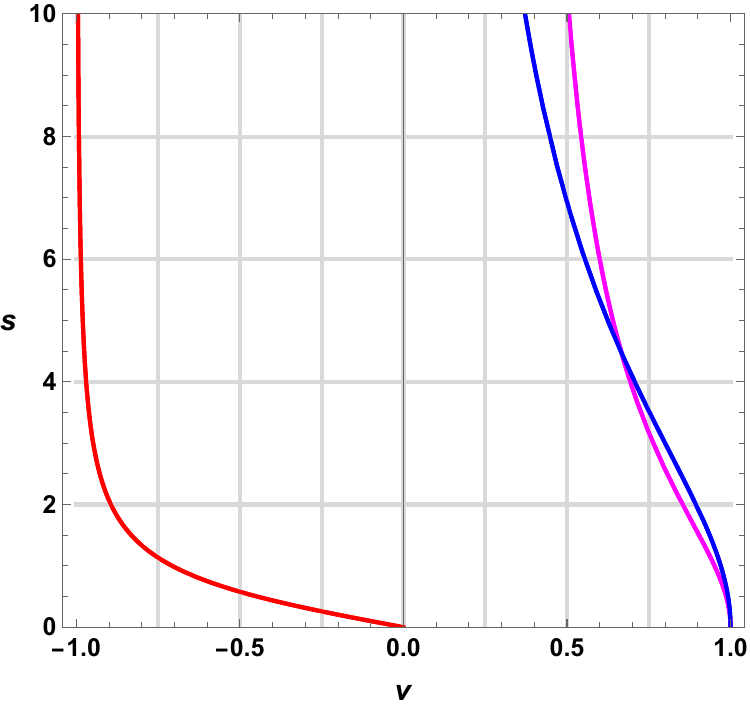}
\includegraphics[width=0.37\linewidth]{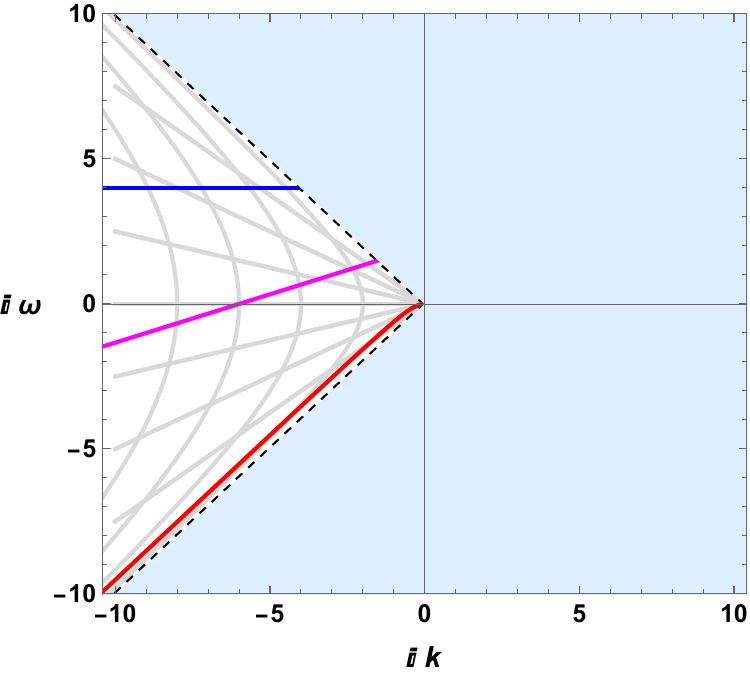}\hspace{0.08\linewidth}
\includegraphics[width=0.35\linewidth]{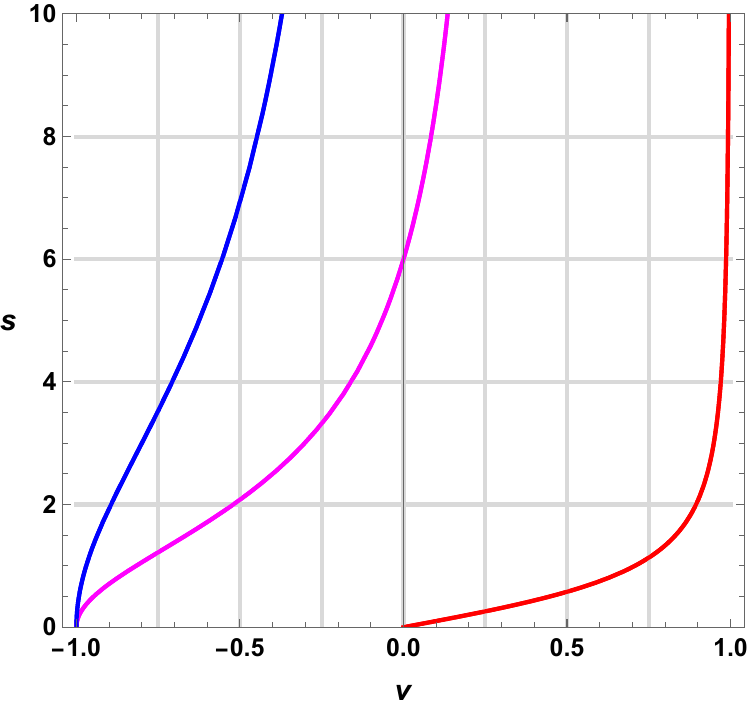}
\caption{Representation of the relaxation-like sector in Milne coordinates (upper panels) and of the evanescent-like sectors in Rindler coordinates (middle and lower panels). For illustration, we plot the dispersion relations of three simple models: a non-propagating relaxation mode, $\partial_t\phi=-4\phi$ (blue); a transport-relaxation model propagating to the right, $\partial_t \phi=-2\phi-\frac{1}{3}\partial_x \phi$ (magenta); and  Cattaneo's theory of causal diffusion, $\partial_t^2\phi+\partial_t\phi=\partial_x^2 \phi$ (red).
In Milne coordinates, the first model is seen to undergo pure Lorentz time dilation, $s\propto 1/\gamma$, reflecting the fact that the corresponding excitation is advected by the medium. By contrast, the telegraph model describes an excitation whose propagation speed approaches the speed of light at large $ik$, leading to an arbitrarily large blueshift as $v\to1$.
The Rindler wedges encode the stationary evanescent modes (wakes and bow waves) generated by obstacles moving at velocity $v$ (see Appendix~\ref{aaa}). The first model exhibits a wake but no bow wave, while the third exhibits a bow wave but no wake. The second model transitions from a bow wave to a wake at $v=1/3$, corresponding to the propagation velocity of the excitation.}
    \label{fig:MilneRindler}
\end{figure}
Here, $v$ is the velocity of the observer in whose frame the mode becomes spatially homogeneous\footnote{Strictly speaking, the standard Milne coordinates are obtained by introducing a rapidity $\eta$ through $v=\tanh\eta$. However, working directly with $v$ is more convenient for our purposes. Similarly for the Rindler parametrization.}, while $s$ is the relaxation rate measured by that observer. This interpretation follows directly from substituting the Milne parametrization into the exponential factor: $e^{ikx-i\omega t}
=
e^{(s\gamma v)x-(s\gamma)t}
=
e^{-s\gamma(t-vx)}
=
e^{-st'}$. Thus, Milne coordinates on the relaxation plane naturally parametrize homogeneous relaxation spectra as the observer is varied (see Fig. \ref{fig:MilneRindler}, upper panels).

Similarly, the evanescent-like sectors can be naturally parametrized using Rindler coordinates \cite{Gron2010}:
\begin{equation}
\begin{cases}
s =\sqrt{(ik)^2-(i\omega)^2} \, , \\
v =i\omega/ik \, ,
\end{cases}
\qquad \Longleftrightarrow \qquad
\begin{cases}
i\omega= \pm s \gamma v \, , \\
ik =\pm s \gamma \, ,
\end{cases}
\end{equation}
where the upper sign corresponds to the right sector and the lower sign to the left sector. Here, $v$ is the velocity of the observer in whose frame the mode becomes stationary, while $s$ determines the spatial decay rate in that frame. In fact, $e^{ikx-i\omega t}
=
e^{(\pm s\gamma)x-(\pm s\gamma v)t}
=
e^{\pm s\gamma(x-vt)}
=
e^{\pm sx'}$. 
Thus, Rindler coordinates naturally parameterize stationary evanescent modes as the observer is varied (see Fig.~\ref{fig:MilneRindler}, middle and lower panels).

\newpage
\section{Causal propagation of the spectrum on the relaxation plane}
\vspace{-0.2cm}

Having developed a global geometric picture of the relaxation plane, we now state our main result. As the parameter $ik$ is varied, the spectrum sweeps out subsets of the relaxation plane. We will show that, in any consistent kinetic theory  (and more generally in theories with analogous mathematical structure, including radiative transfer, transient hydrodynamics, and phenomenological viscoelasticity\footnote{In section \ref{Un9versalityGeometric}, we will show that essentially any causal theory with purely relaxational spectrum can be recast in a kinetic-like form. In other words, the structural assumptions introduced here are equivalent to demanding that $i\omega$ is always real whenever $ik$ is real.}) causality constrains this motion to follow spacelike trajectories. This yields remarkably strong geometric constraints on the mode structure. As an illustration of their consequences, we will then derive rigorous and universal answers to questions \textbf{(i,ii,iii,iv)} posed in the introduction.

\vspace{-0.2cm}
\subsection{Model assumptions}
\vspace{-0.2cm}

Let $\Psi:\text{``Minkowski''}\to \mathcal{H}$ be a linearized classical field describing perturbations, taking values in a complex Hilbert space $\mathcal{H}$, endowed with an inner product $(*,*)$. Physically, $\mathcal{H}$ is the space of all allowed internal states of a material element at a fixed point in spacetime. In transient hydrodynamics, $\Psi$ is a collection of, say, $N$ linearized hydrodynamic variables, so that $\mathcal{H}=\mathbb{C}^N$. In kinetic theory, $\Psi(t,x)$ is the linearized perturbation of the distribution function at the spacetime point $(t,x)$, and therefore $\mathcal{H}=L^2(\mathcal{M})$, where $\mathcal{M}$ denotes the manifold of allowed (quasi)particle momenta. We assume that the dynamics is governed by the following Boltzmann-like equation:
\begin{equation}\label{Boltzmann}
\partial_t\Psi=-\s \Psi-\E \partial_x \Psi \, ,
\end{equation}
where $\s$ and $\E$ are self-adjoint operators on $\mathcal{H}$, with $\E$ bounded. Here, $\s$ encodes equilibration (in kinetic theory, the collision integral), while $\E$ encodes transport (in kinetic theory, the particle velocity operator). Their self-adjointness generically follows from Onsager reciprocity together with $PT$ symmetry \cite{GavassinoNonHydro2022,GavassinoSymmetric2022nff,GavassinoUniveraalityI2023odx,RochaGavassinoFlucut:2024afv,GavassinoConvergence:2024xwf,GavassinoDisturbing:2026klp}, and constitutes the key structural property underlying the Lorentzian geometry of the relaxation plane. In particular, it guarantees that the spectrum is ``purely relaxational'', in the sense that $i\omega\in \mathbb{R}$ whenever $ik\in \mathbb{R}$. In fact, assuming $\Psi\,{\propto}\, e^{ikx-i\omega t}$, equation \eqref{Boltzmann} becomes
\begin{equation}\label{boltzmannFourier}
(\s+ik\E)\Psi=i\omega\Psi \, ,
\end{equation}
so that the set of admissible modes $i\omega$ coincides\footnote{In many contexts, a pair $(i\omega,ik)$ is said to represent an ``excitation mode'' if the Fourier transform of some retarded correlator is singular at $(i\omega,ik)$ \cite{HellerBounds2022ejw}. In Appendix \ref{bbb}, we prove that this condition is equivalent to requiring that $i\omega$ belong to the spectrum of $\s_{ik}$.} with the spectrum of the operator $\s_{ik}=\s+ik\E$. Since $\s_{ik}$ is self-adjoint for $ik\in \mathbb{R}$, we conclude that $i\omega$ is necessarily real.

The boundedness of $\E$ is a consequence of causality. In fact, the spectrum of $\E$ determines the characteristic propagation speeds of the medium \cite{DudynskiEkielJezewska1985,DudyinskiCausality,GavassinoDisturbing:2026klp}, and causality amounts to the statement that the operator norm $w=||\E||$ does not exceed $1$ (in kinetic theory, this corresponds to quasiparticles traveling subluminally). A formal derivation of this operator bound is provided in Appendix \ref{ccc}. 

\vspace{-0.2cm}
\subsection{Main theorems}\label{mainth}
\vspace{-0.2cm}

Consider two wavenumbers $ik_1$ and $ik_2$. The corresponding operators $\s_{ik_1}$ and $\s_{ik_2}$ differ by the bounded self-adjoint operator $(ik_2-ik_1)\E$, whose norm is $|ik_2-ik_1|w$. A standard theorem from the perturbation theory of linear operators \cite[\S 5.4.3., Th. 4.10]{Kato_Perturbation_Theory} then implies that the spectra of $\s_{ik_1}$ and $\s_{ik_2}$ must lie within a distance $|ik_1-ik_2|w$ of one another. Hence, we obtain the following result.

\begin{theorem}\label{theo1}
Let $\mathcal{S}_{ik}=\textup{Spectrum}(\s+ik\E)$ denote the set of modes of the theory \eqref{Boltzmann} at fixed $ik$, and let $w=||\E||\leq 1$ be the maximal propagation speed. Then, for any $ik_1,ik_2\in \mathbb{R}$, we have that $\textup{dist}(\mathcal{S}_{ik_1},\mathcal{S}_{ik_2})\leq |ik_1-ik_2|w$, with $\textup{dist}$ the Hausdorff distance.
In particular, if $i\omega_1\in \mathcal{S}_{ik_1}$, then there exists $i\omega_2\in \mathcal{S}_{ik_2}$ such that
\begin{equation}\label{bound}
|i\omega_1-i\omega_2|\leq w|ik_1-ik_2| \, .
\end{equation}
\end{theorem}

Geometrically, this result implies that, as $ik$ varies, the region of the relaxation plane covered by the spectrum expands at a rate no greater than $w$ (see Fig.~\ref{fig:Propagazia} for representative examples).

The geometric picture can be sharpened even further. Let $i\omega_1$ be an isolated point of $\mathcal{S}_{ik_1}$ with finite degeneracy, and let $d$ denote its distance from the remainder of the spectrum. Analytic perturbation theory \cite{Kato1949_Perturbation_I} implies that, even if splitting were to occur, the resulting spectral cluster remains isolated from the rest of the spectrum throughout the interval $ik\in (ik_1-\frac{d}{2w},ik_1+\frac{d}{2w})$. Within this interval, the cluster there defines a finite collection of dispersion relations $i\omega(ik)$ that are analytic around $ik_1$ (with radius of convergence at least $d/(2w)$ in the non-degenerate case). Specializing \eqref{bound} to the limit $ik_2\to ik_1$, we obtain the following result (see also \cite[\S 7.3.4, Eq. (3.19)]{Kato_Perturbation_Theory}).

\begin{theorem}\label{theo2}
Let $i\omega_1\in \mathcal{S}_{ik_1}$ be a mode of \eqref{Boltzmann} that is separated from the rest of the spectrum at $ik_1\in \mathbb{R}$. If $w=||\E||\leq 1$, then the tangent vector to every dispersion relation $(i\omega(ik),ik)$ crossing $(i\omega_1,ik_1)$ is evanescent-like at $ik_1$. Specifically,
\begin{equation}\label{boundtheo2}
\bigg|\dfrac{d(i\omega)}{d(ik)}\bigg|_{ik=ik_1}\leq w\, .
\end{equation}
\end{theorem}

In the early days of relativity, it was widely expected that causality would manifest itself at the level of dispersion relations through the statement that the group velocity $d\mathfrak{Re}(\omega)/dk$ does not exceed $1$. This expectation was later found to be false in general \cite{Sommerfeld1914,Brillouin1914}, being falsified by simple models such as Cattaneo's theory $\partial_t\phi=(\partial_x^2-\partial_t^2)\phi$ \cite{Pu2010}. We now see that the resolution was simply to rotate the problem into imaginary directions. Indeed, upon introducing an auxiliary variable, Cattaneo's theory can be rewritten in the universal form \eqref{Boltzmann} as $\partial_t\phi=-\partial_x J$ and $\partial_t J=-J-\partial_x\phi$. Theorem \ref{theo2} therefore applies, and bound \eqref{boundtheo2} (with $w{=}1$) is respected, as can be seen from figure \ref{fig:CattAndIS}, left panel.

\begin{figure}[h!]
    \centering
\includegraphics[width=0.36\linewidth]{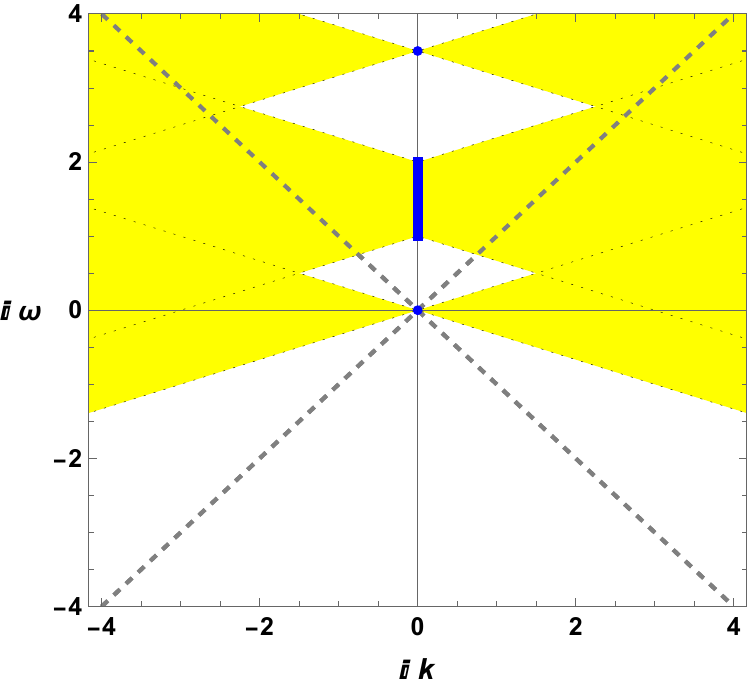}\hspace{0.08\linewidth}
\includegraphics[width=0.36\linewidth]{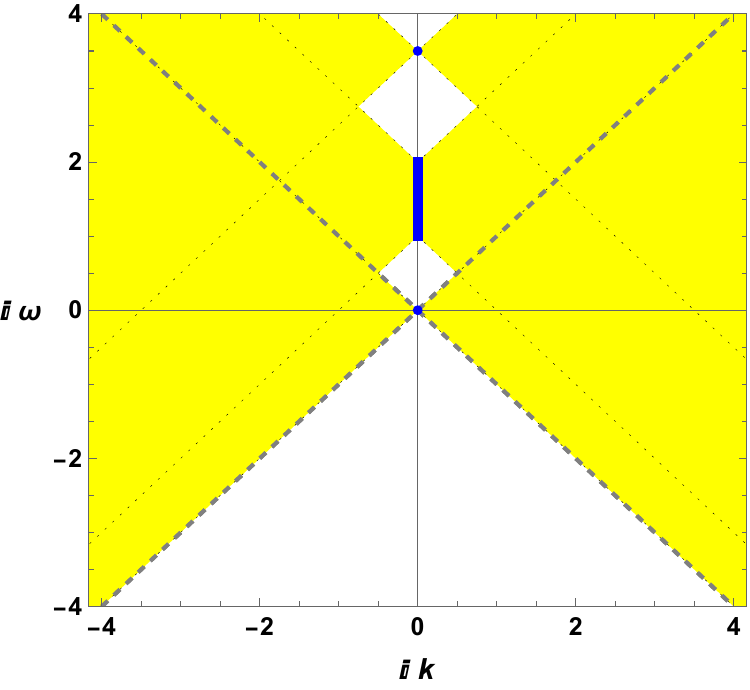}
\includegraphics[width=0.36\linewidth]{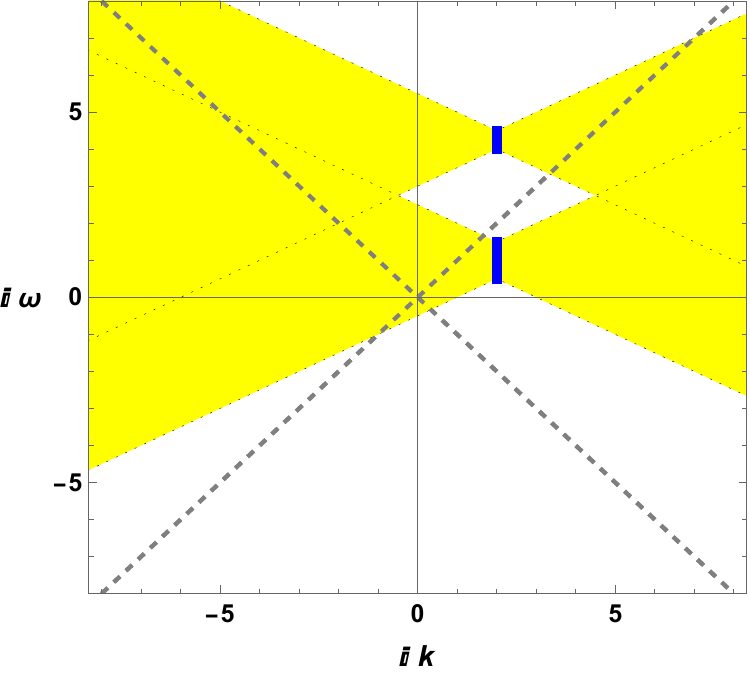}\hspace{0.08\linewidth}
\includegraphics[width=0.36\linewidth]{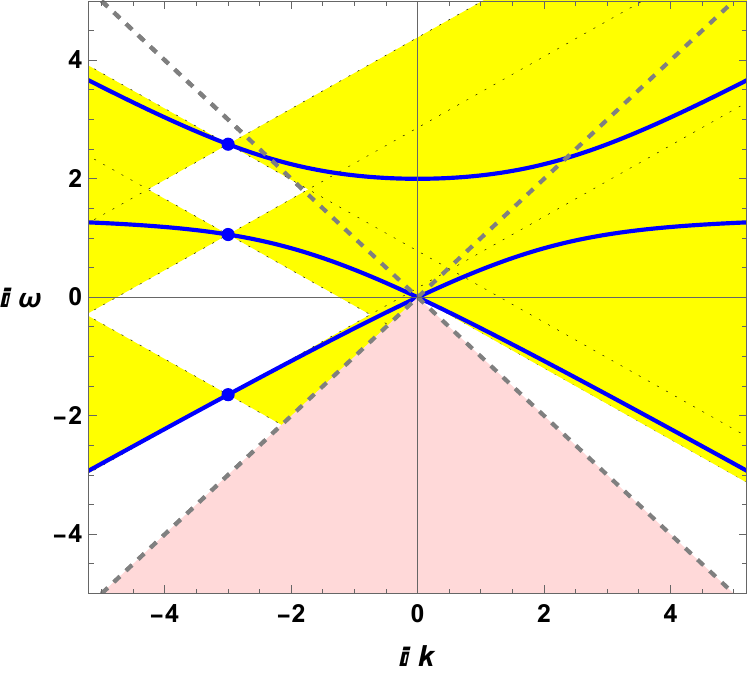}
\caption{Some simple illustrations of how Theorem \ref{theo1} constrains the spectrum to propagate causally on the relaxation plane. 
\textit{Upper left:} Suppose that the spectrum of a kinetic theory is known at $ik=0$. There is a hydrodynamic mode (lower point), a continuous non-hydrodynamic branch (middle band), and a fast-relaxing non-hydrodynamic pole (upper point). If the maximal propagation speed is $w=1/2$, then Theorem \ref{theo1} implies that, as $ik$ varies, the spectrum can expand at most at rate $1/2$. The yellow regions therefore mark the only portions of the relaxation plane where modes may exist.
\textit{Upper right:} Even if the exact propagation speed is not known, causality alone implies $w\leq 1$. The spectrum must therefore remain within cones of slope one.
\textit{Lower left:} The same geometric construction applies when the spectrum is known at nonzero $ik$ (here $ik=2$ and $w=1/2$).
\textit{Lower right:} Application to the model
\(
(\partial_t^2-\tfrac14\partial_x^2)\phi
+\tfrac12\partial_t(\partial_t^2-\tfrac{13}{36}\partial_x^2)\phi=0 ,
\)
which is an Israel--Stewart-type theory of viscous sound, with maximal propagation speed $w=\sqrt{13}/6$. Introducing auxiliary variables, the system can be written in the form \eqref{Boltzmann} as
$\partial_t\phi=-\frac{1}{2}\partial_x J$,
$\partial_t J=-\frac{1}{2}\partial_x\phi-\frac{1}{3}\partial_x\Pi$,
and
$\partial_t\Pi=-2\Pi-\frac{1}{3}\partial_x J$.
This coincides with the standard Israel--Stewart structure upon identifying
$\phi$ with the energy-density perturbation,
$J$ with the flow velocity,
and $\Pi$ with the viscous stress \cite{GavassinoNonHydro2022}.
Assume that one knows three things: \textbf{(a)} the spectrum at $ik=-3$, \textbf{(b)} the value of $w$, and \textbf{(c)} that the theory is stable. Then, the rest of the spectrum can be constrained geometrically by drawing ``cones of light'' emitted from the known spectral points, with opening determined by $w$. In this picture, the unstable-like sector casts a shadow on the allowed spectral region. The exact dispersion relations are shown in blue.}
    \label{fig:Propagazia}
\end{figure}

\subsection{Predicting the breakdown of dispersion relations}
\vspace{-0.3cm}

We now address question \textbf{(i)} posed in the introduction. 
Suppose that we are given a prediction for a dispersion relation $i\omega=f(ik)$ describing an isolated mode, and wish to determine when this prediction ceases to be compatible with causality. Theorem \ref{theo2} immediately provides a simple criterion. As long as the mode remains isolated, its trajectory on the relaxation plane must be evanescent-like. Hence, if $|d(i\omega)/d(ik)|$ becomes larger than $1$, one of two things must occur: either the true mode departs from the predicted curve $f$, meaning that the approximation scheme breaks down, or the mode itself ceases to be isolated. We also note that, since isolated dispersion relations are analytic, if $f$ is exact within some finite interval of $ik$, it remains exact outside that interval, so the mode must collide with another portion of the spectrum before $|d(i\omega)/d(ik)|$ can exceed $1$. Combining this fact with Theorem \ref{theo1}, which constrains how rapidly spectral components may approach one another, allows us to estimate where such collisions must occur.

As a first example, suppose that, in a neighborhood of $ik=0$, the system possesses an isolated mode obeying the diffusion law $i\omega=-\mathfrak{D}(ik)^2$ exactly (with $\mathfrak{D}$ the diffusivity), as in Fokker--Planck kinetic theory in $1{+}1$ dimensions \cite{GavassinoFokkerPLanck:2026zsz}, or in the shear sector of RTA (Relaxation-Time-Approximation) kinetic theory in $2{+}1$ dimensions \cite{Bajec:2024jez}. Such a mode becomes unstable-like for $|ik|>1/\mathfrak{D}$. However, causality forces it to break down much earlier. In fact, applying the bound \eqref{boundtheo2} (with $w\leq1$) to $d(i\omega)/d(ik)=-2\mathfrak{D}ik$, we find that the mode becomes incompatible with causality for $|ik|>1/(2\mathfrak{D})$. Thus, a mode collision must occur no later than $|ik|=(2\mathfrak{D})^{-1}$.
Suppose that, at $ik=0$, the remainder of the spectrum lies within the region $i\omega\geq(4\mathfrak{D})^{-1}$ (which is the case in both Fokker-Planck and RTA). Theorem \ref{theo1} then implies that, at finite $ik$, this portion of the spectrum can move downward at most at unit speed, and must therefore remain within the region $i\omega\geq(4\mathfrak{D})^{-1}-|ik|$. Geometrically, the only possible collision point with the diffusive mode is therefore $(i\omega,ik)=(-(4\mathfrak{D})^{-1},\pm(2\mathfrak{D})^{-1})$, see Fig.~\ref{fig:DeathOfDiffusion} (left panel). Explicit calculations indeed show that this is precisely where the diffusive mode disappears in these models, being absorbed into the continuum.

\begin{figure}[b!]
    \centering
\includegraphics[width=0.39\linewidth]{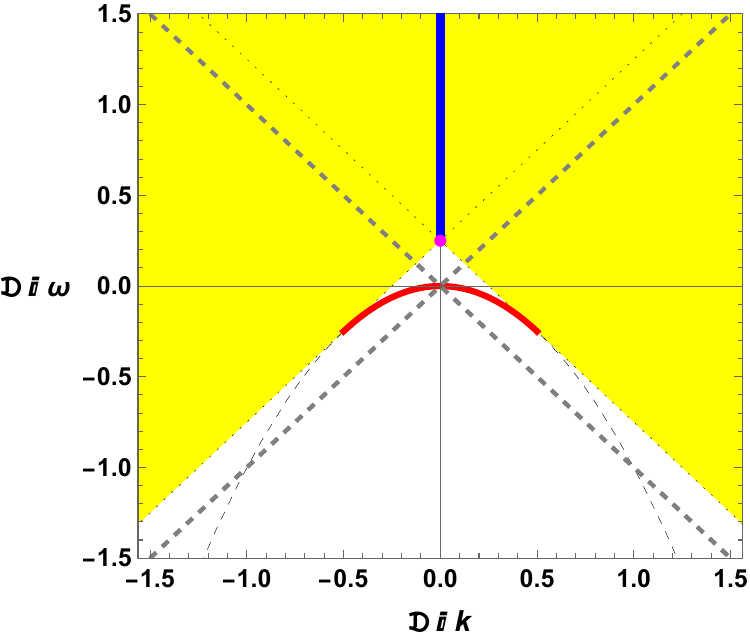}\hspace{0.08\linewidth}
\includegraphics[width=0.39\linewidth]{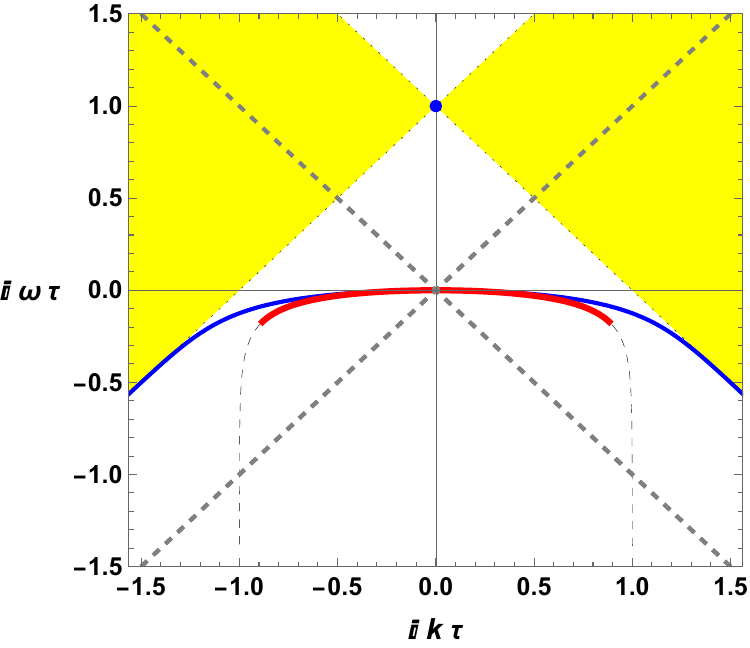}
\caption{Using causality to constrain dispersion relations. 
\textit{Left panel:} Fokker--Planck theory in $1+1$ dimensions and the shear sector of RTA kinetic theory in $2+1$ dimensions possess an isolated diffusive mode obeying $i\omega=-\mathfrak{D}(ik)^2$ exactly (red). For $|\mathfrak{D}ik|>1/2$, the bound \eqref{boundtheo2} is violated, implying that the mode must collide with another spectral component at or before this point. In Fokker--Planck theory, the non-hydrodynamic sector occupies the region $\mathfrak{D}i\omega\geq1/4$ at zero wavenumber (blue). Theorem \ref{theo1} implies that this sector can propagate downward at most at unit speed, forcing a collision precisely at $\mathfrak{D}ik=\pm1/2$, where the diffusive mode is absorbed into the continuum. In RTA kinetic theory, the role of the continuum is played by an infinitely degenerate non-hydrodynamic pole at $\mathfrak{D}i\omega=1/4$ (magenta), which broadens at finite wavenumber into the band $\mathfrak{D}i\omega\in[1/4-|\mathfrak{D}ik|,\,1/4+|\mathfrak{D}ik|]$, again absorbing the diffusive mode at the same point.
\textit{Right panel:} Spiegel's model \eqref{spiegel} for radiative transfer (shown here for $\mathfrak{D}=0.1\tau$) neglects the finite propagation speed of photons. As a result, the predicted mode develops an unphysical divergence at $ik\tau=\pm1$ (dashed). The causal bound \eqref{boundtheo2} predicts that the approximation must fail earlier (transition from red to dashed). The exact dispersion relation is shown in blue.}
    \label{fig:DeathOfDiffusion}
\end{figure}

Let us now consider an example in which the predicted dispersion relation is only approximate. Consider Spiegel's model for radiative transfer in a medium in local thermodynamic equilibrium \cite{Spiegel1957}:
\vspace{-0.2cm}
\begin{equation}\label{spiegel}
i\omega=\dfrac{3\mathfrak{D}}{\tau^2} \left[1-\dfrac{\text{arctanh}(ik\tau)}{ik\tau} \right] \, ,
\end{equation}
where $\mathfrak{D}$ is the radiative diffusivity and $\tau$ is the photon mean free path, with $\mathfrak{D}\ll\tau$. The central approximation underlying this model is that photons propagate instantaneously from their emission point to their absorption point, which manifestly violates causality \cite{GavassinoRadiativeBounds:2025bxx}. Accordingly, Fig.~\ref{fig:DeathOfDiffusion} (right panel) shows that the bound \eqref{boundtheo2} eventually fails. This breakdown occurs precisely when the approximation \eqref{spiegel} ceases to be reliable.

\subsection{Bounds on deviations from time dilation}

We now address question \textbf{(ii)} posed in the introduction. Time dilation states that, if a physical excitation relaxes over a timescale $\tau$ in the medium's rest frame, then the same excitation relaxes over a timescale $\tau'=\gamma\tau$ in a frame moving with velocity $v$. From Fig.~\ref{fig:MilneRindler}, we know that this statement is exact when signals are purely advected by the medium, namely when the transport term $\E$ in \eqref{Boltzmann} vanishes. Once a finite propagation speed $w$ is present, however, deviations from exact time dilation become possible \cite{GavassinoDisturbing:2026klp}. We now derive geometric bounds on such deviations.

Consider an isolated mode located at $(i\omega,ik)=(1/\tau,0)$. Theorem \ref{theo1} implies that, at finite $ik$, this mode must remain within the region
\begin{equation}
\dfrac{1}{\tau}-w|ik|
\leq
i\omega
\leq
\dfrac{1}{\tau}+w|ik| \, .
\end{equation}
Expressing these bounds in Milne coordinates yields $(1+|v|w)^{-1}
\leq
\tau\gamma s
\leq
(1-|v|w)^{-1}$.
Since $s$ represents the relaxation rate $1/\tau'$ of spatially homogeneous modes as measured by an observer moving with velocity $v$, we obtain the following bounds on deviations from time dilation (see Fig.~\ref{fig:DilationFails}):
\begin{equation}\label{bounddilazia}
1-|v|w
\leq
\dfrac{\tau'}{\tau\gamma}
\leq
1+|v|w \, .
\end{equation}

These bounds admit a simple geometric interpretation. Suppose that a signal is emitted at $(t,x)=(0,0)$, propagates with maximal speed $\pm w$, and is absorbed at $(t,x)=(\tau,\pm\tau w)$. In a frame moving with velocity $v$, the same process begins at $(t',x')=(0,0)$ and ends at $(t',x')
=
(\gamma\tau(1\mp vw),\,\gamma\tau(\pm w-v)) $.
Hence, the observed lifetime is bounded between $\tau\gamma(1-|v|w)$ and $\tau\gamma(1+|v|w)$.

The bound \eqref{bounddilazia} is closely related to that derived in \cite{GavassinoDisturbing:2026klp}, but is considerably stronger. First, the upper bound in \eqref{bounddilazia} improves upon that of \cite{GavassinoDisturbing:2026klp}, where only the weaker inequality $\tau'/(\gamma\tau)
\leq
(1-|v|w)^{-1}
=
1+|v|w+|v|^2w^2+\cdots$
was established. Moreover, the result of \cite{GavassinoDisturbing:2026klp} constrains only the extremal points of the non-hydrodynamic spectrum, whereas \eqref{bounddilazia} applies to any isolated mode, and extends naturally to arbitrary isolated components of the spectrum. In particular, suppose that, in the rest frame, an isolated sector of the spectrum is contained within some interval $1/\tau_{\mathrm{Max}}\leq i\omega\leq 1/\tau_{\mathrm{min}}$ at $ik=0$, where $\tau_{\mathrm{Max}}$ and $\tau_{\mathrm{min}}$ denote respectively the longest and shortest relaxation times within that sector. Then, in a boosted frame, the same spectral component must remain within some interval $1/\tau_{\mathrm{Max}}'\leq i\omega'\leq 1/\tau_{\mathrm{min}}'$ (assuming that collisions with other spectral sectors are avoided), with
\begin{equation}\label{taumaxmin}
\tau_{\mathrm{Max}}'
\leq
\gamma(1+|v|w)\tau_{\mathrm{Max}}
\qquad\text{and}\qquad
\tau_{\mathrm{min}}'
\geq
\gamma(1-|v|w)\tau_{\mathrm{min}} \, .
\end{equation}
This follows from the same geometric argument that led to \eqref{bounddilazia}.

\begin{figure}[b!]
    \centering
\includegraphics[width=0.41\linewidth]{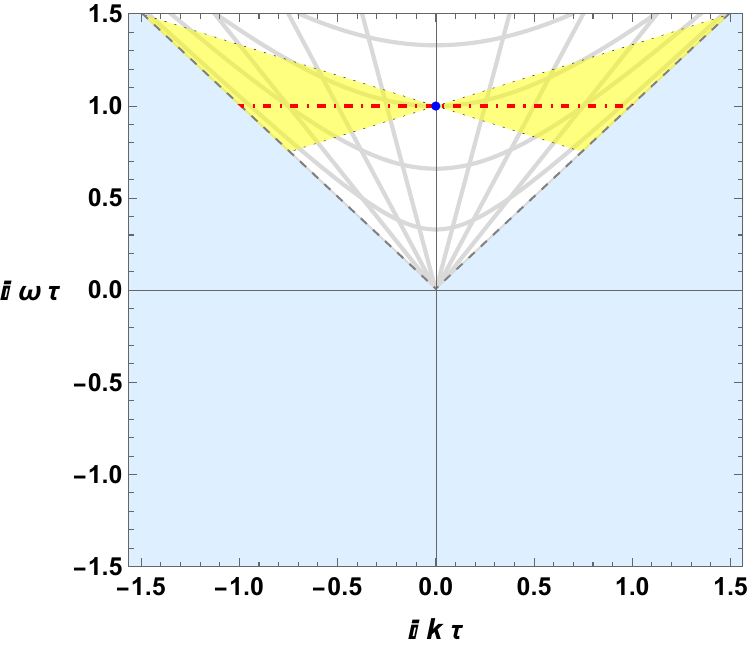}\hspace{0.08\linewidth}
\includegraphics[width=0.4\linewidth]{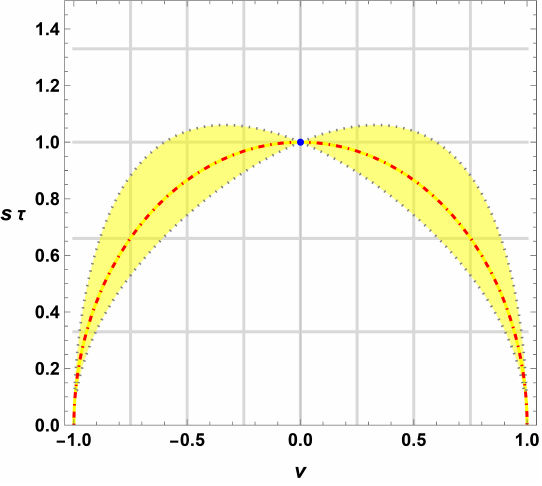}
\caption{Deviations from exact time dilation induced by a finite propagation speed $w$ (here $w=1/3$). We consider an isolated mode satisfying $i\omega\tau=1$ at $ik=0$ (blue point). Exact time dilation, namely $s\tau=1/\gamma$, corresponds to a dispersion relation independent of $ik$ (red line), which is a good approximation in non-relativistic media, where signals effectively comove with the medium (i.e. $w\ll1$) \cite{GavassinoDisturbing:2026klp}. Once the medium supports propagation speeds $w$ comparable to the speed of light, however, the mode may explore a finite region of the relaxation plane (yellow). In Milne coordinates, this translates into a thickened time-dilation profile bounded by \eqref{bounddilazia}.}
    \label{fig:DilationFails}
\end{figure}

\subsection{Breakdown of hierarchies of timescales}\label{hierarchies}

We now address question \textbf{(iii)} posed in the introduction. Suppose that, at $ik=0$, the spectrum splits into a slow sector, whose modes decay over timescales not shorter than $\tau_s$, and a fast sector, whose modes decay over timescales not longer than $\tau_f<\tau_s$. We ask whether this spectral hierarchy must persist in boosted frames.

The answer follows directly from the causal geometry of the relaxation plane (see Fig.~\ref{fig:QuasiMiln}). The slow sector remains confined below the line $i\omega=\tau_s^{-1}+w|ik|$, while the fast sector is confined above the line $i\omega=\tau_f^{-1}-w|ik|$. These two boundaries intersect at the critical wavevectors $(i\omega_c,\pm ik_c)=\frac12\big(\tau_f^{-1}+\tau_s^{-1},\,\pm(\tau_f^{-1}-\tau_s^{-1})/w\big)$.
If these wavevectors are evanescent-like, then the two spectral sectors remain separated throughout the entire relaxation-like region. Passing to Milne coordinates, one concludes that all inertial observers agree on the existence of the hierarchy. Hence,
\begin{equation}\label{ObsIndepHier}
w
\leq
\dfrac{\tau_s-\tau_f}{\tau_s+\tau_f}
\qquad\Longrightarrow\qquad
\text{the spectral hierarchy is observer-independent.}
\end{equation}
Conversely, if the critical wavevectors are relaxation-like, then the two causal regions overlap inside the relaxation sector. In Milne coordinates, this implies that there exist boosted frames in which the spectral separation may disappear. One therefore obtains
\begin{equation}
w
>
\dfrac{\tau_s-\tau_f}{\tau_s+\tau_f}
\qquad\Longrightarrow\qquad
\text{possible breakdown for }
|v|
\geq
\dfrac{\tau_s-\tau_f}{w(\tau_s+\tau_f)} \, .
\end{equation}

Several immediate consequences follow. If $w=1$, the only hierarchy guaranteed to remain observer independent is the limiting case $\tau_s=\infty$ with finite $\tau_f$, corresponding to the separation between hydrodynamic modes and a gapped non-hydrodynamic sector. In other words, causality alone guarantees only the persistence of the hierarchy between hydrodynamic behavior and microscopic equilibration. By contrast, if $\tau_s\gg\tau_f$ but finite, guaranteeing \textit{a priori} that the hierarchy is maintained in all boosted frames requires $w$ to be strictly smaller than $1$. Finally, in non-relativistic media ($w\ll1$), any initially separated spectral sectors remain separated for all observers.

We can make a further interesting observation. Suppose that the critical wavevectors $(i\omega_c,\pm ik_c)$ are evanescent-like, so we are in the case \eqref{ObsIndepHier}. Then the two spectral sectors cannot intersect within the relaxation-like region, but may still collide inside the Rindler wedges. Suppose that such a collision occurs precisely at $(i\omega_c,\pm ik_c)$. Passing to Rindler coordinates, one finds that this intersection controls a transition in the structure of non-equilibrium wakes (see Appendix \ref{aaa}). In particular, an immersed object moving at speed $v>0$ leaves behind a non-equilibrium tail whose structure changes at $v=|i\omega_c/ik_c|$: below this speed, the wake is controlled by one continuous set of decay lengths, whereas above it the wake separates into two distinct spatial-decay components.

\begin{figure}[h!]
    \centering
\includegraphics[width=0.40\linewidth]{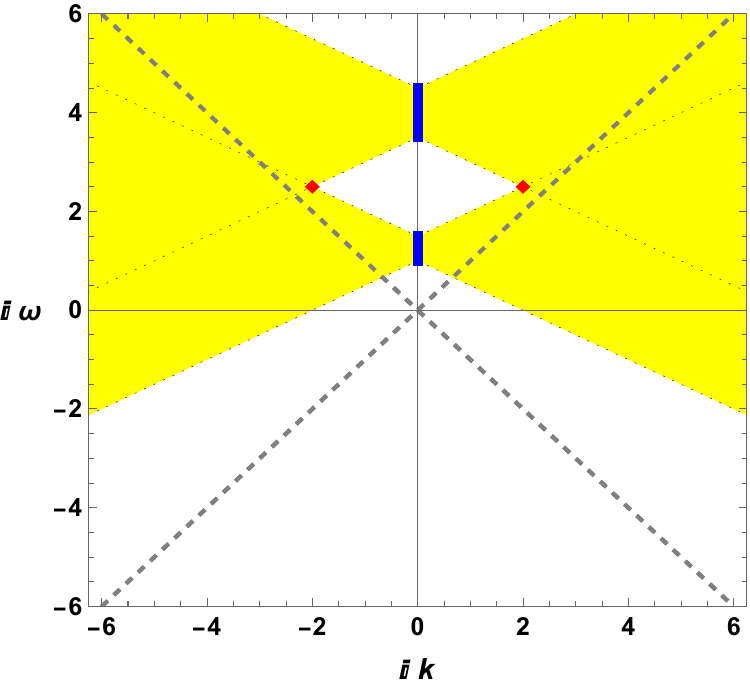}\hspace{0.08\linewidth}
\includegraphics[width=0.39\linewidth]{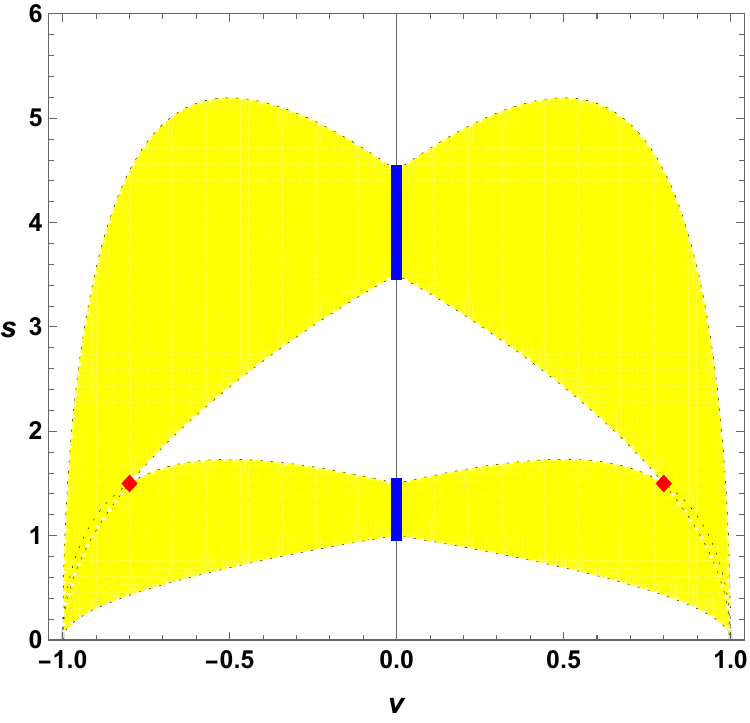}
\caption{An example of a model in which a rest-frame spectral hierarchy breaks down at large boosts. At $ik=0$, two disconnected spectral components occupy the intervals $i\omega\in[1,1.5]$ and $i\omega\in[3.5,4.5]$ (blue segments). The maximal signaling speed is assumed to be $w=1/2$. Propagating these spectral regions at the maximal allowed rate $1/2$, one finds that they can meet at the critical wavevectors $(i\omega_c,\pm ik_c)=(2.5,\pm2)$ (red dots), which lie inside the relaxation-like sector. Passing to Milne coordinates (right panel), one concludes that observers moving with speeds $|v|\geq|ik_c/i\omega_c|=0.8$ may observe a single connected spectral component (at vanishing wavenumber).}
    \label{fig:QuasiMiln}
\end{figure}

\subsection{The regime of validity of hydrodynamics}
\vspace{-0.3cm}

Using the insight developed in the previous subsections, we are now ready to address question \textbf{(iv)}. Suppose that, in the rest frame of the medium and at $ik=0$, there is a (possibly finitely degenerate) mode with $i\omega=0$, while the rest of the spectrum is confined to the region $i\omega\geq 1/\tau_g$, where $\tau_g$ may be interpreted as the non-hydrodynamic gap. As $ik$ is turned on, the zero mode generally splits into a finite collection of hydrodynamic modes, each analytic in a neighborhood of $ik=0$ \cite{Dudynski1989Hydro,GavassinoConvergence:2024xwf}. Standard analytic perturbation theory \cite{Kato1949_Perturbation_I} guarantees that these hydrodynamic modes remain analytic for real $ik$ at least up to $|ik|=(2w\tau_g)^{-1}$, even if they happen to cross one another at finite $ik$ \cite[\S 7.3.2]{Kato_Perturbation_Theory}. Moreover, $(2w\tau_g)^{-1}$ is also the minimal value of $|ik|$ at which the hydrodynamic and non-hydrodynamic sectors can collide. Hence, whether one defines the radius of validity of hydrodynamics as the distance from the origin within which hydrodynamic dispersion relations are analytic, or as the range of $ik$ within which hydrodynamic modes remain the longest-lived excitations, one obtains
\vspace{-0.1cm}
\begin{equation}\label{Rbound}
\mathcal{R}\geq \dfrac{1}{2w\tau_g} \, .
\end{equation}

Let us now discuss what happens in boosted frames. Substituting \eqref{boost} into \eqref{boltzmannFourier}, we obtain $[\s+ik'\gamma(\E-v)]\Psi=i\omega'\gamma(1-v\E)\Psi$. Since $1-v\E$ is self-adjoint and positive definite by causality, we can define
\begin{equation}
\Psi' =\gamma^{1/2} (1-v\E)^{1/2}\Psi \, , \qquad
\s' =\gamma^{-1} (1-v\E)^{-1/2}\s (1-v\E)^{-1/2}\, , \qquad
\E' = (1-v\E)^{-1} (\E-v) \, , 
\end{equation}
and thereby obtain, in the boosted frame as well, a Boltzmann-like spectral problem of the form $(\s'+ik'\E')\Psi'=i\omega'\Psi'$. Thus, the arguments of Sec.~\ref{mainth} on existence and analyticity of dispersion relations apply unchanged. In particular, by Lorentz-transforming the relaxation plane (see Fig.~\ref{fig:HydroBreaks}), one finds that collisions with the non-hydrodynamic spectrum are necessarily avoided as long as
\vspace{-0.3cm}
\begin{equation}
\dfrac{\gamma (-1-vw)}{2w\tau_g} < ik'< \dfrac{\gamma (1-vw)}{2w\tau_g} \, ,
\end{equation}
so the hydrodynamic modes exist and remain analytic throughout this interval.
The convergence radius $\mathcal{R}'$ of the Taylor expansion of an individual boosted dispersion relation $i\omega'(ik')$ around $ik'{=}0$ can, however, be smaller. It is bounded (modulo hydrodynamic-mode collisions) by the boosted analogue of the 
estimate \eqref{Rbound}, $\mathcal{R}'\,{\geq}\, (2w'\tau_g')^{-1}$ \cite{Kato1949_Perturbation_I}, where $w'\equiv||\E'||\leq (w+|v|)/(1+|v|w)$ by the spectral mapping theorem, and $\tau_g'\leq \gamma(1+|v|w)\tau_g$ by \eqref{taumaxmin}. Thus, 
\begin{equation}
\mathcal{R}' \geq \dfrac{1}{2\gamma(w+|v|)\tau_g} \, .
\end{equation}
 
\begin{figure}[h!]
    \centering
\includegraphics[width=0.39\linewidth]{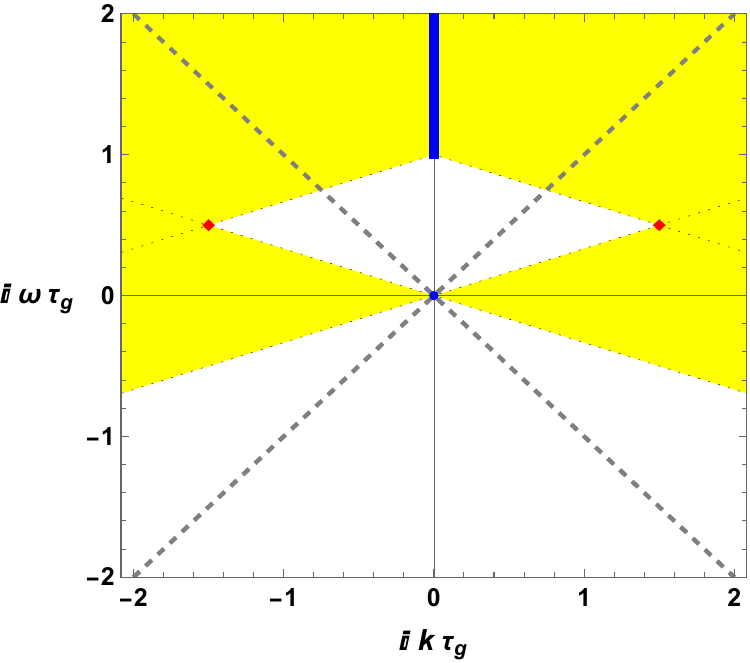}\hspace{0.08\linewidth}
\includegraphics[width=0.39\linewidth]{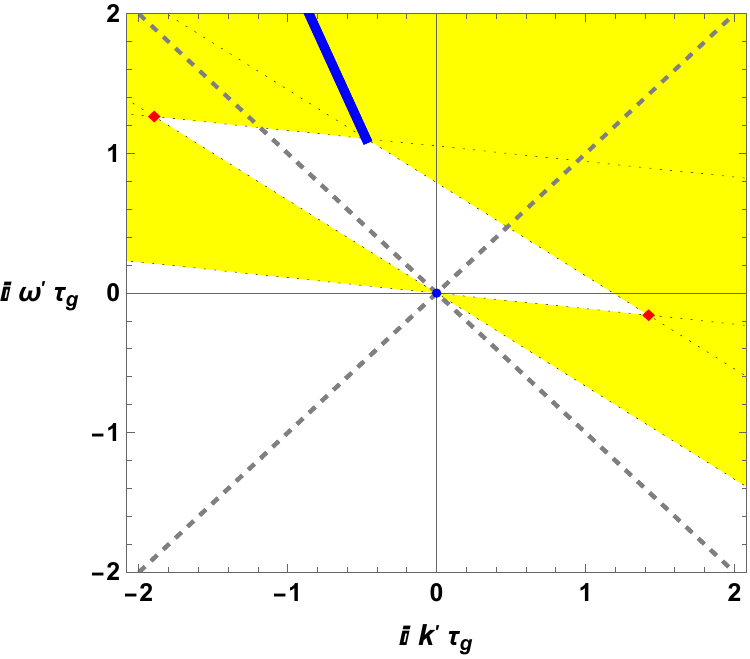}
\caption{The regime of validity of hydrodynamics. \textit{Left panel:} In the rest frame, the hydrodynamic modes can collide with the non-hydrodynamic spectrum no earlier than at the critical wavevectors $(i\omega_c,\pm ik_c)=(\frac{1}{2\tau_g},\pm\frac{1}{2w\tau_g})$ (red dots). Hence, the hydrodynamic dispersion relations necessarily exist and remain analytic at least within the interval $ik\in[-ik_c,ik_c]$. \textit{Right panel:} After a Lorentz boost, the same existence and analyticity arguments continue to apply. However, the critical wavevectors become asymmetric, with $ik_c'=\frac{\gamma}{2w\tau_g}(\pm1-vw)$. Consequently, the hydrodynamic regime is distorted on the relaxation plane.}
    \label{fig:HydroBreaks}
\end{figure}

\section{Morphology of the hydrodynamic sector}

In Fig.~\ref{fig:CattAndIS}, we displayed the spectra of a minimal theory of diffusion (left panel) and a minimal theory of sound (right panel), both of which can be recast in the standard form \eqref{Boltzmann}. In these examples, the non-hydrodynamic sector is of Israel--Stewart type, consisting of a single non-hydrodynamic mode, and is therefore much simpler than that of a realistic kinetic theory. Nevertheless, one can still ask which qualitative features of the \textit{hydrodynamic} part of the spectrum survive in more complicated models. Since hydrodynamic modes remain evanescent-like as long as they are separated from the rest of the spectrum, such an analysis directly constrains the structure of wakes and bow waves in the long-wavelength regime (see Appendix \ref{aaa}).

In the following, we will focus on stable theories with a gapped non-hydrodynamic sector, namely theories satisfying $\text{Spectrum}(\s)\subseteq \{0\}\cup [1/\tau_g,\infty)$ with finite $\tau_g$.

\subsection{Diffusion of a single charge}\label{diffusionofasignlecharge}

Let us consider a system with a single relevant conservation law (typically energy or particle number). Then, the mode $(i\omega,ik)=(0,0)$ is non-degenerate, and turning on a finite $ik$ produces a single hydrodynamic mode. As long as this mode remains isolated, it is analytic. Furthermore, standard analytic perturbation theory \cite{Kato1949_Perturbation_I} yields
\begin{equation}\label{seconderivative}
\dfrac{d^2(i\omega)}{d(ik)^2}
=
-2\left(\mathbb{P}\E\Psi,(\s_{ik}-i\omega)^{-1}_{\Psi^\perp}\mathbb{P}\E\Psi \right) \, ,
\end{equation}
where $\Psi$ is the normalized eigenvector associated with the mode, $(\s_{ik}-i\omega)^{-1}_{\Psi^\perp}$ denotes the resolvent restricted to the Hilbert subspace $\Psi^\perp$ orthogonal to $\Psi$, and $\mathbb{P}$ is the orthogonal projector onto $\Psi^\perp$.

Now, since the diffusive mode is the lowest eigenvalue of $\s_{ik}$, the operator $(\s_{ik}-i\omega)^{-1}_{\Psi^\perp}$ is non-negative definite. Hence, $(i\omega)''\leq0$, meaning that the diffusive mode is always downward concave.
Suppose moreover that the system is invariant under the reflection $x\mapsto -x$ (i.e. there is a unitary transformation that sends $\s\mapsto \s$ and $\E\mapsto -\E$). Since the hydrodynamic mode is unique, it must then be symmetric under $ik\to-ik$, and therefore satisfy $i\omega'(0)=0$. Combined with downward concavity, this implies that $i\omega'\leq0$ for positive $ik$ and $i\omega'\geq0$ for negative $ik$. Integrating, one concludes that $i\omega\leq0$ everywhere. All these features are explicitly visible in Figs.~\ref{fig:CattAndIS} (left panel) and \ref{fig:DeathOfDiffusion}.

Finally, at $ik=0$, we have $||(\s_{ik}-i\omega)^{-1}_{\Psi^\perp}||=1/\text{distance}(0,\text{Spectrum}(\s)\backslash\{0\})\leq\tau_g$. Therefore,
\begin{equation}\label{estimiamo}
\bigg|
\dfrac12\dfrac{d^2(i\omega)}{d(ik)^2}
\bigg|_{ik=0}
=
|\left(\mathbb{P}\E\Psi,(\s_{ik}-i\omega)^{-1}_{\Psi^\perp}\mathbb{P}\E\Psi \right)|
\leq
||(\s_{ik}-i\omega)^{-1}_{\Psi^\perp}||
\,||\E||^2
\leq
w^2\tau_g \, .
\end{equation}

Combining these observations, we obtain the following result.

\begin{theorem}\label{theo3}
Consider a system of the form \eqref{Boltzmann} with a single conserved charge. Suppose that the system is invariant under the reflection $x \mapsto -x$. Then, as long as the charge-diffusion mode $i\omega(ik)$ remains isolated, the following properties hold for real $ik$: \textup{\textbf{(a)}} $i\omega\leq0$, \textup{\textbf{(b)}} $(i\omega)''\leq0$, and \textup{\textbf{(c)}} $i\omega(ik)$ is even in $ik$. Moreover, Taylor-expanding around the origin yields $i\omega=-\mathfrak{D}(ik)^2+\cdots$, where the diffusivity satisfies the universal bound
\begin{equation}\label{Dwtau}
0\leq \mathfrak{D}\leq w^2\tau_g \, .
\end{equation}
\end{theorem}
The causal bound $\mathfrak{D}\lesssim w^2\tau_g$ on diffusivity was first proposed in \cite{Hartman:2017hhp} using a simple qualitative argument. A particle undergoing Brownian motion travels a typical distance $x=\sqrt{\mathfrak{D}t}$ over a time $t$. At sufficiently early times, this exceeds the maximal causal distance $x=wt$. Ordinary diffusion must therefore break down near the timescale $\tau\sim\mathfrak{D}/w^2$, where the two curves intersect. This is precisely the scale at which one expects non-hydrodynamic modes to appear. The remarkable fact is that, for theories governed by \eqref{Boltzmann}, the above causal bound is \textit{exact}.
In fact, it is also optimal, since it is saturated by Cattaneo's theory, corresponding to $\mathcal{H}=\mathbb{C}^2$ and
\begin{equation}
\E=
\begin{bmatrix}
0 & w \\
w & 0 \\
\end{bmatrix}
\, ,
\qquad\qquad
\s=
\begin{bmatrix}
0 & 0 \\
0 & 1/\tau_g \\
\end{bmatrix} \, .
\end{equation}
Also, note that \eqref{Dwtau} is considerably tighter than the corresponding hydrohedron bound $\mathfrak{D}\leq 16 w/(3\pi \mathcal{R})$. Indeed, combining this hydrohedron inequality with \eqref{Rbound}, one obtains the weaker estimate
\begin{equation}\label{hydrohedronbound}
\mathfrak{D}
\leq
\dfrac{32}{3\pi} w^2\tau_g
\approx
3.4\, w^2\tau_g \qquad\qquad (\text{hydrohedron bound}) \, .
\end{equation}
On the other hand, it should be emphasized that the bound \eqref{Dwtau} need not apply to theories whose spectrum is oscillatory at real $ik$, and which therefore cannot be recast in the Boltzmann-like form \eqref{Boltzmann}. By contrast, the hydrohedron bound remains valid in such cases.

\subsection{Sound waves}
\vspace{-0.3cm}

Let us now consider a fluid at zero chemical potential. Suppose that this fluid is isolated, so that energy and momentum conservation must be explicitly enforced at the level of \eqref{Boltzmann}. Moreover, assume that the system is in a normal phase, so that no additional Goldstone modes are present.\footnote{Superfluids can also be described within theories of the form \eqref{Boltzmann} \cite{GavassinoUniveraalityI2023odx}. Their characteristic signature would then be the doubling of the sound modes, corresponding to first and second sound.}

Restricting attention to excitations that are isotropic in the directions transverse to the $x$ axis (i.e. ``spin-0'' states), one finds that the mode $(i\omega,ik)=(0,0)$ is twofold degenerate, due to energy and longitudinal momentum conservation. Turning on a finite $ik$ lifts this degeneracy and produces two distinct hydrodynamic modes: the left-moving and right-moving sound waves. If the system is invariant under the reflection $x\mapsto-x$, these two modes must be specular copies of one another. Expanding around the origin, one then finds $i\omega=\pm c_s ik-\mathfrak{D}(ik)^2+\cdots$,
where $c_s\in[0,w]$ is the speed of sound and $\mathfrak{D}$ is the acoustic diffusivity.

Away from the origin, the two sound modes are generically non-degenerate, and \eqref{seconderivative} applies. In particular, the branch with lower $i\omega$ is always downward concave, since $(\s_{ik}-i\omega)^{-1}_{\Psi^\perp}$ is non-negative definite. No analogous general statement exists for the upper branch at finite $ik$. 
Near $ik=0$, however, equation \eqref{seconderivative} is no longer useful, because the two modes become arbitrarily close, and the reduced resolvent diverges. Indeed, at $ik=0$, the identity itself ceases to apply, since it is valid only for non-degenerate modes. In this regime, one must instead use degenerate perturbation theory \cite{Kato1950_Perturbation_II}. Let $\Psi_\pm$ denote the orthonormal basis that diagonalizes the restriction of $\E$ to $\ker(\s)$, and let $\mathbb{P}$ be the orthogonal projector onto $\{\Psi_+,\Psi_-\}^\perp$. Choosing $\Psi_+$ to represent the right-moving sound mode, so that $(\Psi_+,\E\Psi_+)=+c_s$, and using $\mathbb{P}\Psi_+=0$, we obtain
\begin{equation}\label{estimiamo2}
\begin{split}
\mathfrak{D}
& =
\left(\mathbb{P}\E\Psi_+,(\s_{ik}-i\omega)^{-1}_{\{\Psi_+,\Psi_-\}^\perp}\mathbb{P}\E\Psi_+ \right)
\leq
\tau_g \left(\mathbb{P}\E\Psi_+,\mathbb{P}\E\Psi_+ \right)
\\
& =
\tau_g \left(\mathbb{P}(\E-c_s)\Psi_+,\mathbb{P}(\E-c_s)\Psi_+ \right)
\leq
\tau_g \left((\E-c_s)\Psi_+,(\E-c_s)\Psi_+ \right)
\\
& =
\tau_g \left[(\Psi_+,\E^2\Psi_+)-c_s^2\right]
\leq
\tau_g(w^2-c_s^2) \, .
\end{split}
\end{equation}
The quantity $w^2-c_s^2$ is just the variance of the propagation operator $\E$ within the sound eigenstate $\Psi_+$, meaning that diffusion is controlled by the incoherent fluctuations of the microscopic propagation speeds around the collective sound motion.
Combining these observations, we obtain the following result.

\begin{theorem}\label{theo4}
Consider a fluid governed by \eqref{Boltzmann}, whose only longitudinal hydrodynamic excitations are two sound modes $i\omega_\pm(ik)$. Suppose that the system is invariant under the reflection $x\mapsto-x$, so that $i\omega_-(ik)=i\omega_+(-ik)$. Then, as long as the sound modes remain isolated from the non-hydrodynamic spectrum, the following properties hold for real $ik$: \textup{\textbf{(a)}} $i\omega_\pm$ remain analytic even if they cross one another at finite $ik$, and \textup{\textbf{(b)}} at every $ik$, the branch with lower $i\omega$ satisfies $(i\omega)''\leq0$. Moreover, Taylor-expanding around the origin yields $i\omega=\pm c_s(ik)-\mathfrak{D}(ik)^2+\cdots$, where
\begin{equation}\label{boundonsound}
0\leq c_s\leq w \, ,
\qquad\qquad
0\leq \mathfrak{D}\leq (w^2-c_s^2)\tau_g \, .
\end{equation}
\end{theorem}

Also in this case, the upper bound on $\mathfrak{D}$ is tighter\footnote{The physical interpretation of this bound parallels the argument of \cite{Hartman:2017hhp}. Consider a sound pulse whose centroid propagates at speed $c_s$. In the diffusive regime, its spatial variance grows as $\langle x^2\rangle-\langle x\rangle^2\sim\mathfrak{D}t$, so that $\langle x^2\rangle\sim c_s^2t^2+\mathfrak{D}t$. Causality requires $\langle x^2\rangle\leq w^2t^2$, implying that acoustic diffusion must break down on a timescale $\tau\sim\mathfrak{D}/(w^2-c_s^2)$.} than the corresponding hydrohedron bound. Moreover, it is optimal, being saturated by the following model (with $\mathcal{H}=\mathbb{C}^4$):
\begin{equation}\label{maximalDiff}
\E=
\begin{bmatrix}
c_s & (w^2{-}c_s^2)^{1/2} & 0 & 0 \\
(w^2{-}c_s^2)^{1/2} & -c_s & 0 & 0 \\
0 & 0 & -c_s & -(w^2{-}c_s^2)^{1/2} \\
0 & 0 & -(w^2{-}c_s^2)^{1/2} & c_s \\
\end{bmatrix} \, ,
\qquad\qquad
\s =
\begin{bmatrix}
0 & 0 & 0 & 0 \\
0 & 1/\tau_g & 0 & 0 \\
0 & 0 & 0 & 0 \\
0 & 0 & 0 & 1/\tau_g \\
\end{bmatrix} \, .
\end{equation}

Matching the expansion $i\omega=\pm c_s(ik)-\mathfrak{D}(ik)^2+\cdots$ to the corresponding hydrodynamic dispersion relation allows one to rewrite the acoustic diffusivity in terms of the shear and bulk viscosities $\eta$ and $\zeta$. One then obtains
\begin{equation}\label{soundboundwow}
c_s^2+
\dfrac{\zeta+\frac43\eta}{2(\varepsilon+P)\tau_g}
\leq
w^2
\leq
1 \, ,
\end{equation}
where $\varepsilon{+}P$ is the enthalpy density. Causal bounds relating first-order transport coefficients to the speed of sound had previously been proposed based on Israel--Stewart theory \cite{Hippert:2024hum}, but no rigorous result had been established for general kinetic-type theories. Notably, the constraint obtained from the characteristic analysis of Israel--Stewart is too strong:
\begin{equation}
c_s^2+
\dfrac{\zeta+\frac43\eta}{(\varepsilon+P)\tau_g}
\leq
w^2
\leq
1
\qquad\qquad
(\text{Israel--Stewart bound}) \, .
\end{equation}

In Fig.~\ref{fig:MaximalSound}, we compare the spectrum of the maximally viscous model \eqref{maximalDiff} with that of Israel--Stewart theory.

\begin{figure}[h!]
    \centering
\includegraphics[width=0.39\linewidth]{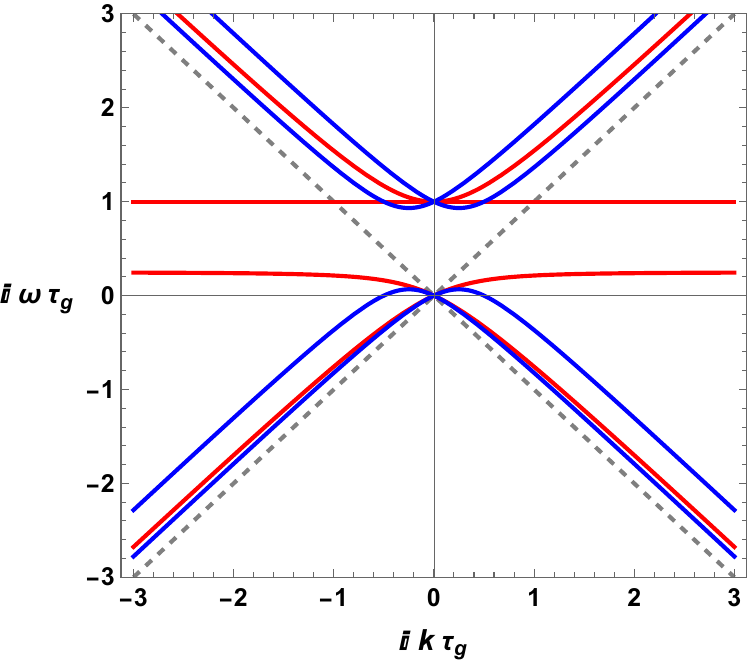}\hspace{0.08\linewidth}
\includegraphics[width=0.39\linewidth]{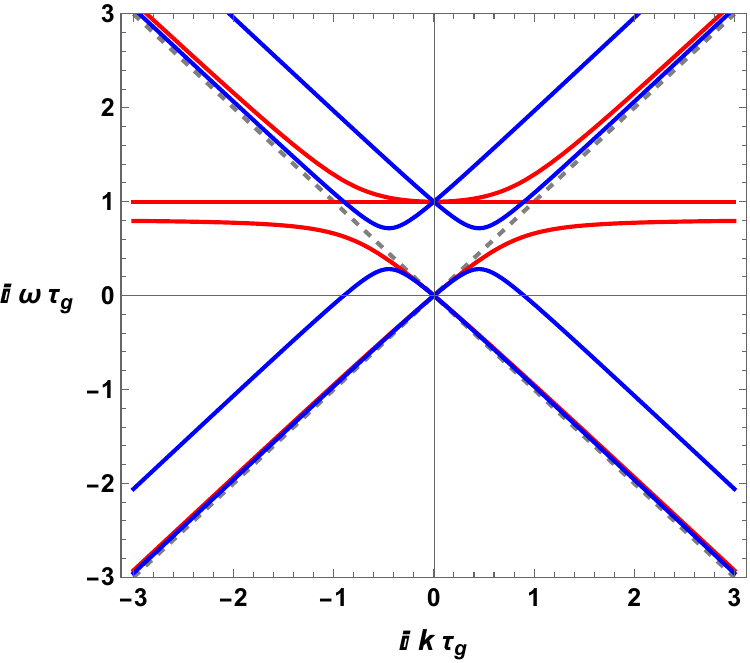}
\caption{Comparison between the sound sector of the maximally viscous model \eqref{maximalDiff} (blue) and Israel--Stewart theory with bulk and shear viscosity (red), for $\{w,c_s\}=\{1,0.5\}$ (left panel) and $\{1,0.9\}$ (right panel). For Israel--Stewart theory, we set both relaxation times equal to $\tau_g$. With this choice, the spectrum is insensitive to how the total diffusivity is distributed between bulk and shear viscosity at fixed $w$, and one always obtains $\mathfrak{D}=\frac12(w^2-c_s^2)\tau_g$. Thus, Israel--Stewart theory realizes only half of the maximal diffusivity allowed by causality, whereas the model \eqref{maximalDiff} saturates the bound $\mathfrak{D}\leq (w^2-c_s^2)\tau_g$ exactly. The model \eqref{maximalDiff} is maximally viscous because the operator $\E$ has only the eigenvalues $\pm w$, so $\langle \E^2\rangle=w^2$. Hence, the sound eigenstates have the maximal possible velocity variance, $\langle \E^2\rangle-\langle \E\rangle^2=w^2-c_s^2$, and therefore undergo maximal spatial spreading.}
    \label{fig:MaximalSound}
\end{figure}

\subsection{Shear and tensor waves}
\vspace{-0.3cm}

In an ordinary medium in $3+1$ dimensions, the total Hilbert space $\mathcal{H}$ naturally decomposes into spin sectors (sometimes called channels) corresponding to excitations that transform differently under rotations around the $x$ axis:
\begin{equation}
\mathcal{H}
=
\bigoplus_{m=-\infty}^{\infty}\mathcal{H}_m \, ,
\qquad\qquad
\left(
\mathcal{H}_m
\to
e^{im\varphi}\mathcal{H}_m
\text{ under a rotation by angle }\varphi
\right) \, .
\end{equation}
If a state $\Psi$ belongs to a given channel, the associated perturbation $\delta T^{\mu\nu}$ of the energy-momentum tensor must transform with the corresponding phase under rotations in the $yz$ plane. This enforces the following algebraic structure:
\begin{equation}\label{sectorialization}
\begin{split}
m=0
\qquad \Longrightarrow \qquad
&
\delta T^{\mu\nu} =
\begin{bmatrix}
A & B & 0 & 0 \\
B & C & 0 & 0 \\
0 & 0 & D & 0 \\
0 & 0 & 0 & D \\
\end{bmatrix}
\, ,
\\
m=\pm1
\qquad \Longrightarrow \qquad
&
\delta T^{\mu\nu} =
\begin{bmatrix}
0 & 0 & A & \pm iA \\
0 & 0 & B & \pm iB \\
A & B & 0 & 0 \\
\pm iA & \pm iB & 0 & 0 \\
\end{bmatrix}
\, ,
\\
m=\pm2
\qquad \Longrightarrow \qquad
&
\delta T^{\mu\nu} =
\begin{bmatrix}
0 & 0 & 0 & 0 \\
0 & 0 & 0 & 0 \\
0 & 0 & A & \pm iA \\
0 & 0 & \pm iA & -A \\
\end{bmatrix}
\, .
\\
\end{split}
\end{equation}
From this decomposition, one sees that sound waves belong to $\mathcal{H}_0$, shear waves belong to $\mathcal{H}_{\pm1}$, and tensor waves belong to $\mathcal{H}_{\pm2}$. Higher-spin sectors contain excitations with no hydrodynamic counterpart, since $\delta T^{\mu\nu}=0$ for $|m|>2$.

In a fluid, there is one single hydrodynamic mode within each sector $\mathcal{H}_{\pm1}$, arising from the conservation of transverse momentum. Thus, the mode $(i\omega,ik)=(0,0)$ is non-degenerate within these sectors. Assuming invariance under the reflection $x\mapsto-x$, we are back in the setting of Sec.~\ref{diffusionofasignlecharge}, and Theorem \ref{theo3} applies to shear waves as well. By contrast, in a solid, which can also be described using theories of the form \eqref{Boltzmann} \cite{GavassinoUniveraalityI2023odx}, additional conservation laws associated with elastic memory give rise to new hydrodynamic excitations. In particular, there are now two hydrodynamic modes within each sector $\mathcal{H}_{\pm1}$. Consequently, shear waves in solids satisfy the properties stated in Theorem \ref{theo4}.

The tensor sectors $\mathcal{H}_{\pm2}$ cannot contain hydrodynamic modes in fluids, since no relevant conservation law exists there. Usually, solids do not contain hydrodynamic tensor waves either, because the relevant hydrodynamic degree of freedom is the displacement of the volume elements, which transforms as a spin-1 quantity under rotations.

\subsection{Bounds on the shear viscosity}
\vspace{-0.2cm}

Applying the bound \eqref{Dwtau} to the shear channel of relativistic fluids yields an upper bound on the shear viscosity:
\begin{equation}\label{etaoverepsplusp}
\dfrac{\eta}{\varepsilon+P}
\leq
w^2\tau_g
\qquad\qquad
(\text{for fluids}) \, ,
\end{equation}
which is saturated by Israel--Stewart theory. Here, $\tau_g$ may be identified with the non-hydrodynamic gap of $\s$ restricted to $\mathcal{H}_{\pm1}$, and is, in general, different from the $\tau_g$ entering \eqref{soundboundwow}, which corresponds instead to the gap in the sound channel $\mathcal{H}_0$.

We also note that additional information about the microscopic kinematics of the theory can considerably refine the bound \eqref{etaoverepsplusp} through \eqref{estimiamo}. Consider, for example, an ideal non-degenerate ultrarelativistic gas described within kinetic theory. Then, $\mathcal{H}=L^2(\mathbb{R}^3)$, $\E=p^x/p^0$ (so $w=1$), and $\Psi(p^x,p^y,p^z)$ is related to the perturbation of the distribution function through $f=f_{\text{eq}}+\sqrt{f_{\text{eq}}}\Psi$ \cite{DudynskiEkielJezewska1985}. A hydrodynamic shear wave with helicity $+1$ at $ik=0$ is then represented by $\Psi=\sqrt{f_{\text{eq}}}(p^y+ip^z)$. Taking into account that this state is not normalized, we obtain
\begin{equation}
\begin{split}
\dfrac{\eta}{\varepsilon+P}
&=
\dfrac{\left(\mathbb{P}\E\Psi,(\s_{ik}-i\omega)^{-1}_{\Psi^\perp}\mathbb{P}\E\Psi \right)}{(\Psi,\Psi)}
\leq
\tau_g
\dfrac{||\E\Psi||^2}{||\Psi||^2}
\\
&=
\tau_g
\dfrac{||\sqrt{f_\text{eq}}(p^y+ip^z)p^x/p^0||^2}
{||\sqrt{f_\text{eq}}(p^y+ip^z)||^2}
\\
&=
\tau_g
\dfrac{\int_{4\pi} |\Omega^y+i\Omega^z|^2 (\Omega^x)^2 d^2\Omega}
{\int_{4\pi} |\Omega^y+i\Omega^z|^2 d^2\Omega}
=
\dfrac{\tau_g}{5}
\qquad\qquad
(\text{for ideal ultrarelativistic gases}) \, .
\\
\end{split}
\end{equation}
This kinetic upper bound was first derived in \cite{Ghiglieri:2018dgf}. Here, we see that it arises naturally as a particular instance of a more general bound on the concavity of dispersion relations on the relaxation plane.

An interesting consequence of the bound \eqref{etaoverepsplusp} is the following. Consider a fluid with vanishing chemical potential, so that $\varepsilon+P=Ts$, where $T$ is the temperature and $s$ the entropy density. Then, assuming the Kovtun--Son--Starinets (KSS) bound \cite{KovtunSonStarinets:2004de}, $\eta/s\geq 1/(4\pi)$, and combining it with \eqref{etaoverepsplusp} together with causality ($w\leq1$), one obtains a lower bound on the first non-hydrodynamic relaxation time:
\begin{equation}
\tau_g
\geq
\dfrac{1}{4\pi T}
\qquad\qquad
(\text{for KSS-respecting fluids}) \, .
\end{equation}
It should be emphasized, however, that this bound applies only to systems with purely relaxational spectra, since it relies on the representation \eqref{Boltzmann}. Systems close to saturating the KSS bound are often described holographically, where the non-hydrodynamic spectrum is typically oscillatory \cite{KovtunHolography2005,Heller2014}, and the above inequality need not apply. Nevertheless, we are not aware of explicit counterexamples in fluids.

\vspace{-0.2cm}
\section{Morphology of the non-hydrodynamic sector}
\vspace{-0.2cm}

A longstanding question in relativistic hydrodynamics is the formulation of general criteria to determine whether the spectrum of a given system consists only of discrete eigenvalues (``poles'') or also contains continuous branches (``cuts'') \cite{Moore:2018mma,Kurkela:2017xis,RochaBranchcut:2024cge,GavassinoGapless:2024rck}. In kinetic theory, this question is nontrivial even in the simplest models (such as RTA) where a discrete non-hydrodynamic mode at $ik=0$ can open into a continuous branch at finite $ik$, signaling ballistic propagation \cite{RomatschkeCutsandPoles:2015gic,Bajec:2025dqm}. In relativistic systems, this phenomenon is especially significant because a branch that opens at finite $ik$ in the rest frame also opens at finite $v$ in Milne coordinates (see Fig.~\ref{fig:DilationFails}). In this situation, a mode that appears pole-like at zero wavenumber in one frame may already correspond to a cut at zero wavenumber in another.

In this section, we derive general results on the structure of the non-hydrodynamic sector.

\vspace{-0.2cm}
\subsection{A no-go theorem for cuts at finite wavenumber}
\vspace{-0.2cm}

We have the following no-go result for the existence of a continuous ballistic branch at finite wavenumber.

\begin{theorem}\label{theo5}
Consider a system governed by \eqref{Boltzmann}. Suppose that, at $ik=0$, the spectrum consists only of discrete eigenvalues with finite multiplicity, and has no finite accumulation point. Then, the same is true for all $ik\in\mathbb{C}$.
\end{theorem}
\begin{proof}
The assumptions are equivalent to the statement that $\s$ has compact resolvent. Since $\E$ is bounded, the operators $\s_{ik}$ form a holomorphic family of type (A). Moreover, each $\s_{ik}$ is closed and has non-empty resolvent set by \cite[\S 5.4.3, Problem 4.8]{Kato_Perturbation_Theory}. Hence, by \cite[\S 7.2.1, Theorem 2.4]{Kato_Perturbation_Theory}, every operator $\s_{ik}$ also has compact resolvent.
\end{proof}

\subsection{Necessity of ballistic cuts in RTA-type theories}
\vspace{-0.3cm}

RTA kinetic theory escapes the assumptions of Theorem \ref{theo5} (see Fig.~\ref{fig:DeathOfDiffusion}) because the RTA non-hydrodynamic mode $(i\omega,ik)=(1/\tau,0)$ is an eigenvalue with infinite multiplicity. Indeed, an infinitely degenerate relaxation mode necessarily opens into a ballistic cut whenever the propagation operator has continuous essential spectrum.

\begin{theorem}\label{theo6}
Consider a system governed by \eqref{Boltzmann}. Suppose that, at $ik=0$, the spectrum consists of a finite number of discrete eigenvalues, exactly one of which, denoted by $\lambda$, has infinite multiplicity. Suppose moreover that the operator $\E$ has essential spectrum $[-w,w]$. Then, for finite $ik\in\mathbb{C}$, the eigenvalue $\lambda$ opens into the ballistic branch $\lambda+ik[-w,w]$.
\end{theorem}
\begin{proof}
Under the above assumptions, $\s$ differs from $\lambda$ by a finite-rank operator: $\s=\lambda+\mathbb{K}$, with $\mathbb{K}$ compact. Hence, $\s_{ik}=\lambda+ik\E+\mathbb{K}$. Since compact perturbations do not change the essential spectrum \cite[\S 4.5.6, Theorem 5.35]{Kato_Perturbation_Theory}, the essential spectrum of $\s_{ik}$ coincides with that of $\lambda+ik\E$, namely $\lambda+ik[-w,w]$.
\end{proof}

To apply this result to relativistic RTA kinetic theory, one only needs to observe that $\E=p^1/p^0$ is a multiplication operator. Its essential spectrum is $[-1,1]$, giving $\text{Ess-Spectrum}(\s_{ik})=1/\tau+ik[-1,1]$ (this is true also for $ik$ complex).

\vspace{-0.4cm}
\subsection{The non-hydrodynamic sector of the Boltzmann equation}
\vspace{-0.3cm}

Inspired by the discussion above, we now prove a result tailored specifically to kinetic-theory applications.

\begin{theorem}\label{theo7}
Consider a system governed by \eqref{Boltzmann}, with $\mathcal{H}=L^2(\mathcal{M})$ for some smooth manifold $\mathcal{M}$. Suppose that $\s=\mathbb{L}+\mathbb{K}$, where $\mathbb{L}$ is the multiplication operator associated with a function $\ell\in L^\infty(\mathcal{M})$, and $\mathbb{K}$ is a self-adjoint compact operator. Suppose furthermore that $\E$ is the multiplication operator associated with a function $e\in L^\infty(\mathcal{M})$. Then, for every $ik\in\mathbb{C}$, the essential range of the function $\ell+ike$ belongs to the excitation spectrum.
\end{theorem}

\begin{proof}
Since $\mathbb{K}$ is compact, the essential spectrum of $\s_{ik}$ coincides with that of $\mathbb{L}+ik\E$. The latter, being a multiplication operator, has spectrum equal to the essential range of its associated function \cite[\S IX.2.6]{conway1990functional}.
\end{proof}

To appreciate the significance of this result, consider an ideal relativistic gas in $3+1$ dimensions with perturbed distribution function $f=f_\text{eq}+\sqrt{f_\text{eq}}\,\Psi$. Then, $\mathcal{H}=L^2(\mathbb{R}^3)$ and $\E=p^1/p^0$ is a multiplication operator. Moreover, it is well known \cite{DudynskiEkielJezewska1985} that Boltzmann's $2$-to-$2$ collision operator can always be written in the form $\s=\tau(p^0)^{-1}+\mathbb{K}$, where $\mathbb{K}$ is compact and $\tau(p^0)$ is a (generically energy-dependent) relaxation time \cite{DudynskiEkielJezewska1985}. Applying Theorem \ref{theo7}, we immediately conclude (see also \cite{Dudynski1989Hydro}) that the whole region
\vspace{-0.2cm}
\begin{equation}
i\omega
\in
\bigcup_{p^0}
\left\{
\tau(p^0)^{-1}
\right\}
+
ik[-1,1]
\end{equation}
belongs to the excitation spectrum.
Thus, at $ik=0$, the system generically possesses a cut spanning a continuum of relaxation times. For real $ik$, this cut simply expands at speed $1$ across the relaxation plane. For real $k$, instead, it becomes a two-dimensional non-analyticity region in the complex $\omega$ plane extending over the interval $\mathfrak{Re}\omega\in[-k,k]$. This shows that the Boltzmann collision kernel generically admits ballistic propagation at all velocities. Of course, this does not exhaust the full spectrum, since the compact operator $\mathbb{K}$ may generate additional discrete eigenvalues.

\vspace{-0.5cm}
\section{A case study}
\vspace{-0.3cm}

We conclude the list of applications by discussing a relativistic kinetic theory in $2+1$ dimensions that can be solved analytically, yet behaves very differently from standard RTA-like and Boltzmann-like models, in that it possesses no ballistic branch cut. This concrete example will illustrate how an $\{i\omega,ik\}$-plane analysis can be used to extract nontrivial information about the linear-response properties of relativistic matter.

\vspace{-0.4cm}
\subsection{Physical setup}
\vspace{-0.3cm}

We consider a gas of photons with fixed energy propagating on a plane. Their velocity may be parameterized as $v^j=(\cos\theta,\sin\theta)$. Suppose that these photons undergo elastic scattering with a medium, so that their energy is conserved, while their propagation angle $\theta$ changes over time. In the limit where the scattering processes become infinitely soft but infinitely frequent, the angle $\theta$ of each photon undergoes Brownian motion, and the Boltzmann equation for the radiative intensity $I(t,x,y,\theta)$ reads\footnote{In this setup, the radiative intensity $I$ is simply the number of photons per unit area and unit angle, $dN/(dx\,dy\,d\theta)$. The left-hand side of \eqref{booltzmannPhotons} is then just Liouville's theorem. In fact, suppose that no collisions occur. Then, in a time interval $\Delta t$, a free-streaming photon located at $\Gamma=(x,y,\theta)$ in phase space moves to the point $\Gamma(\Delta t)=(x+\cos\theta\,\Delta t,y+\sin\theta\,\Delta t,\theta)$. The vector field generating this flow, $\Dot{\Gamma}=(\cos\theta,\sin\theta,0)$, is divergence-free. Hence, the phase-space conservation law $\partial_t I+\partial_a(\Dot{\Gamma}^a I)=0$ reduces to $\partial_t I+\Dot{\Gamma}^a\partial_a I=0$.}
\begin{equation}\label{booltzmannPhotons}
\partial_t I
+\cos(\theta)\partial_x I
+\sin(\theta)\partial_y I
=
\dfrac{1}{\tau}\partial_\theta^2 I \, ,
\end{equation}
where $\tau$ is a characteristic scattering timescale.

\newpage
In this setting, the Hilbert space is naturally $\mathcal{H}=L^2(\text{circle})$, equipped with inner product $(I,J)=\int_{-\pi}^{\pi}\frac{d\theta}{2\pi}I^*J$. Taking the inner product of \eqref{booltzmannPhotons} with either $1$ or $I$, one obtains respectively the conservation of the number of photons and a free-energy dissipation law:
\begin{equation}
\begin{split}
&\partial_t (1,I)
+\partial_x(\cos\theta,I)
+\partial_y(\sin\theta,I)
=
0 \, ,
\\
&\partial_t \dfrac{(I,I)}{2}
+\partial_x \dfrac{(I,\cos\theta I)}{2}
+\partial_y \dfrac{(I,\sin\theta I)}{2}
=
-\dfrac{1}{\tau}(\partial_\theta I,\partial_\theta I)
\leq
0 \, .
\\
\end{split}
\end{equation}
These identities immediately imply that the system possesses at least one hydrodynamic mode and no unstable-like modes. Notably, the hydrodynamic mode is necessarily unique. Indeed, the free-energy law shows that any state with $(i\omega,ik)=(0,0)$ must satisfy $(\partial_\theta I,\partial_\theta I)=0$, which implies that $I$ is constant in $\theta$. Hence, the mode $(i\omega,ik)=(0,0)$ is non-degenerate. Combined with isotropy, this places the system precisely in the setting of Theorem \ref{theo3}: the photons undergo diffusion.

\vspace{-0.1cm}
\subsection{Determination of the spectrum}
\vspace{-0.1cm}

Comparing \eqref{booltzmannPhotons} with \eqref{Boltzmann}, we immediately identify $\s=-\tau^{-1}\partial_\theta^2$ and $\E=\cos(\theta)$. Both operators are manifestly self-adjoint, and $\E$ has norm $w=1$. Determining the spectrum therefore amounts to solving
\begin{equation}
-\partial_\theta^2 I
+
ik\tau \cos(\theta) I
=
i\omega\tau I \, ,
\qquad\qquad
\text{with }
\{\theta=\pi\}\sim\{\theta=-\pi\} \, .
\end{equation}
This is precisely the Schr\"odinger equation for a particle confined to the unit circle in the $xy$ plane, with kinetic energy $L_z^2$ and potential energy proportional to $x$. At $ik=0$, the eigenvalues are $i\omega_n\tau=n^2$ ($n=0,1,2,\ldots$, so $\tau_g=\tau$), with eigenfunctions $\cos(n\theta)$ and $\sin(n\theta)$. Hence, all excited states (i.e. all non-hydrodynamic modes) are twofold degenerate.
Turning on a finite $ik$ lifts this degeneracy (note that Theorem \ref{theo5} applies), yielding
\begin{equation}
\begin{cases}
i\omega_n^{c}\tau=\frac{1}{4} a_{2n}(2ik\tau) \, , \\
I_n^c (\theta) =\text{ce}_{2n}(\theta/2,2ik\tau) \, ,
\end{cases}
\qquad \qquad
\begin{cases}
i\omega_n^{s}\tau=\frac{1}{4} b_{2n}(2ik\tau) \, , \\
I_n^s (\theta) =\text{se}_{2n}(\theta/2,2ik\tau) \, ,
\end{cases}
\end{equation}
where $\text{ce}_{2n}$ and $\text{se}_{2n}$ are respectively the even and odd Mathieu functions of the first kind, while $a_{2n}$ and $b_{2n}$ are the associated Mathieu characteristic numbers. Note that the mode $n=0$ exists only in the cosine family, whereas the sine family starts at $n=1$. The resulting spectrum is shown in Fig.~\ref{fig:PhotonsPLane}.

\begin{figure}[h!]
    \centering
\includegraphics[width=0.38\linewidth]{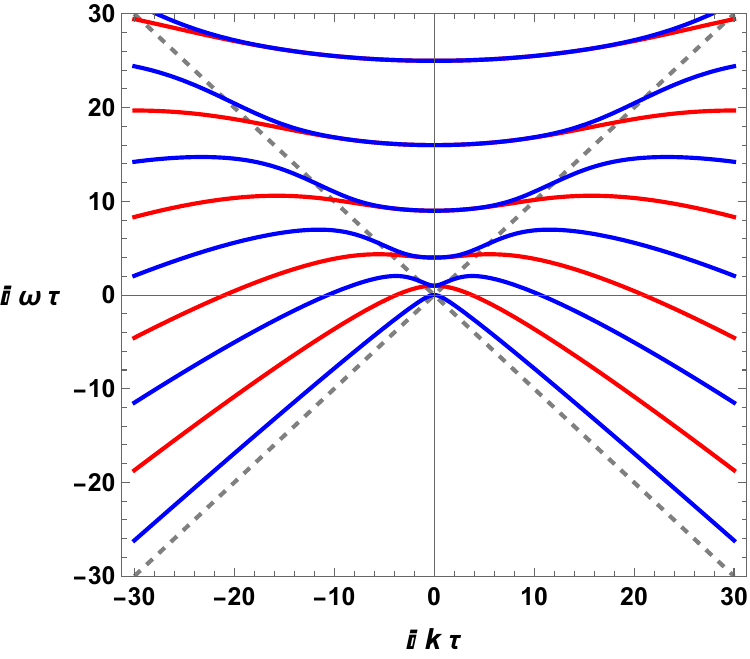}\hspace{0.08\linewidth}
\includegraphics[width=0.38\linewidth]{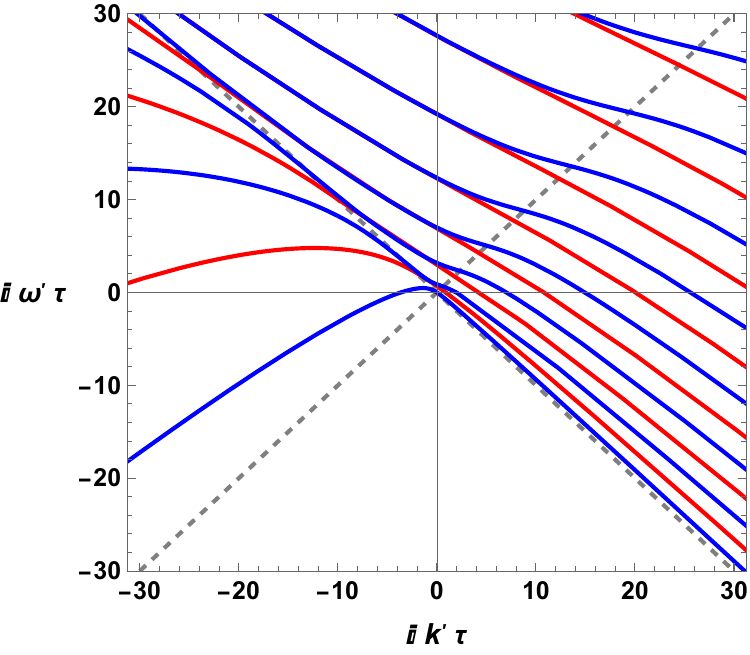}
\caption{Dispersion relations of the model \eqref{booltzmannPhotons}, describing photons propagating on a plane and undergoing elastic scattering. \textit{Left panel}: Excitations in the rest frame. The system possesses a single diffusive mode, together with an infinite tower of non-hydrodynamic modes that are twofold degenerate at $ik=0$, and split at finite $ik$ into two families: one even under $\theta\mapsto-\theta$ (blue), and one odd (red). The absence of a continuous branch of the form $i\omega\in \text{const}+[-ik,ik]$ is a consequence of Theorem \ref{theo5}, and indicates that this kinetic model does not possess the ballistic spectral structure characteristic of most kinetic theories. \textit{Right panel}: Spectrum in a reference frame moving with velocity $v=0.7$. Examining the modes with $i\omega'=0$, we find that the hydrodynamic bow wave (i.e. the point where the hydrodynamic mode intersects the $ik'$ axis to the left of the origin) lies farther from the origin than the first non-hydrodynamic wake. Hence, sufficiently fast perturbing objects do not experience a clean separation between hydrodynamic and non-hydrodynamic scales.}
    \label{fig:PhotonsPLane}
\end{figure}

\subsection{Small- and large-wavenumber limits}

Since the hydrodynamic dispersion relation is known analytically, all transport coefficients can be obtained by Taylor-expanding around the origin:
\begin{equation} i\omega^c_0\tau= -\frac{1}{2} (ik\tau)^2+\frac{7}{32} (ik\tau)^4-\frac{29 }{144} (ik\tau)^6+\frac{68687}{294912} (ik\tau)^8-\frac{123707}{409600} (ik\tau)^{10}+... \, . \end{equation}
Hence, the diffusion coefficient is $\mathfrak{D}=\tau/2$, precisely half of the maximal allowed value, while the super-Burnett coefficient is $\mathfrak{B}=7\tau^3/32$, and so on \cite{StruchtrupSuperBurnett2005,Struchtrup2011Review}. The non-hydrodynamic dispersion relations can likewise be expanded in Taylor series, although their coefficients do not admit a simple interpretation in terms of hydrodynamic transport.

The large-$ik$ behavior of the spectrum follows from a simple quantum-mechanical argument. As $ik\tau<0$ grows in magnitude, the potential $ik\tau\cos(\theta)$ becomes very deep, and the wavefunctions localize near $\theta=0$. Expanding around this point gives the harmonic-oscillator Hamiltonian $-\partial_\theta^2+|ik\tau|(\theta^2/2-1)$, and therefore
\begin{equation}
i\omega_n^{c,s}\tau \approx -|ik\tau|+\sqrt{|2ik\tau|} (2n\pm 1/2) \, , 
\end{equation}
which we have also verified numerically. 

\subsection{Application: Wake behind a luminal pulse}

Let us now discuss a simple physical phenomenon that can be analyzed directly on the $\{i\omega,ik\}$ plane. Consider an impulsive force front propagating along the lightlike surface $x=t$. This may represent, for example, an electromagnetic or gravitational wave. Such a pulse perturbs the medium and generates a cloud of photons on the front, which then diffuse behind it. Discarding transient effects, we may therefore adopt the ansatz $I(t,x,\theta)
=
\text{const}
+
\delta I(x-t,\theta)\Theta(t-x)$,
where $\delta I(x-t,\theta)$ solves \eqref{booltzmannPhotons} in the region $x<t$, and $\delta I(0,\theta)$ specifies the distribution of photons generated at the pulse front.
Assuming equilibrium at $x=-\infty$, together with invariance under reflections of the $y$ axis, the general solution takes the form
\begin{equation}\label{deltaIII}
\delta I
=
\sum_{n=0}^\infty
c_n I_n^c(\theta)e^{ik_n(x-t)}
\qquad\qquad
(\text{for }x<t) \, ,
\end{equation}
where the $ik_n$ are the solutions of the equation $ik_n=i\omega_n^c$, namely the wavenumbers at which the blue curves in Fig.~\ref{fig:PhotonsPLane} intersect the diagonal. The coefficients $c_n$ are fixed by the boundary conditions at the pulse front. 

As a simple model of forward emission, we take $c_n\,{=}\,1/(\pi\sqrt{2})$.
The resulting series converges for $x\,{<}\,0$, and describes a profile in which the photons are emitted approximately comoving with the front (i.e. with $\theta\sim0$), see Fig.~\ref{fig:PhotonsIsotorpization}.

\begin{figure}[h!]
    \centering
\includegraphics[width=0.42\linewidth]{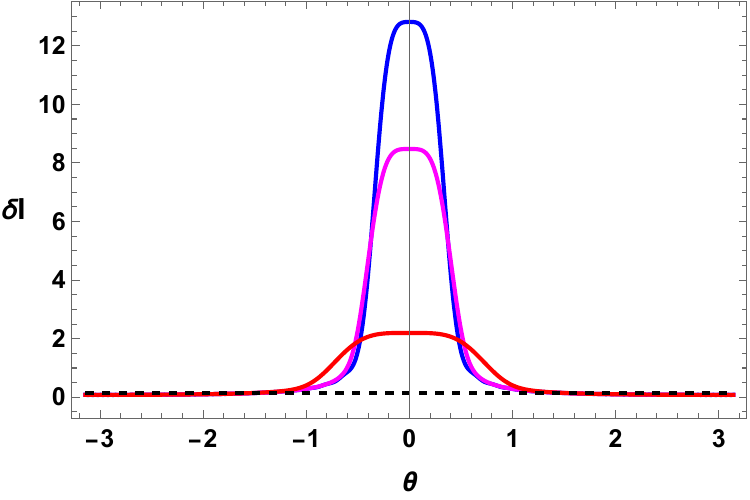}\hspace{0.08\linewidth}
\includegraphics[width=0.43\linewidth]{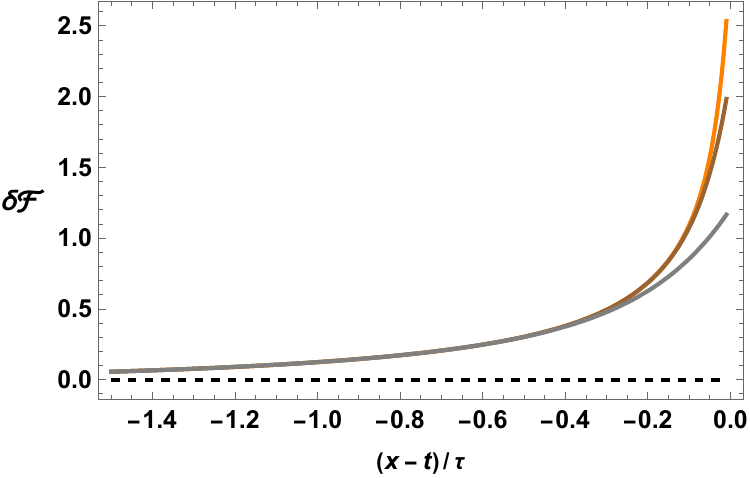}
\caption{Wake generated behind a pulse that perturbs a photon distribution governed by \eqref{booltzmannPhotons}. \textit{Left panel:} Photon distribution \eqref{deltaIII} for $c_n=1/(\pi\sqrt{2})$ (numerically truncated at $n_\text{max}=250$) at the sample points $(x-t)/\tau=-0.0005$ (blue), $-0.001$ (magenta), $-0.01$ (red), and $-\infty$ (dashed). Near the front, the photons propagate almost entirely toward positive $x$, and progressively isotropize farther behind it. At infinity, the distribution approaches a constant isotropic state with perturbed photon density $(1,\delta I)=1$, meaning that the pulse has injected photons into the system. \textit{Right panel:} Profile of the photon flux $\delta\mathcal{F}=(\cos\theta,\delta I)$ for different truncations: $n_\text{max}=0$ (dashed), $2$ (gray), $4$ (brown), and $6$ (orange). Higher-order terms decay on progressively shorter length scales, so including additional modes enhances the flux increasingly close to the front.}
    \label{fig:PhotonsIsotorpization}
\end{figure}

\section{Universality of the geometric bounds}\label{Un9versalityGeometric}

Our proofs and derivations made explicit use of the Boltzmann-like representation \eqref{Boltzmann}. On the other hand, the resulting geometric structures appear disentangled from the underlying dynamical variables. Instead, they seem to depend primarily on two main assumptions: the existence of a maximal signaling speed $w\leq 1$, and the reality of $i\omega$ for real $ik$. This suggests that the Lorentzian geometry uncovered in the previous sections may be more fundamental than the specific representation \eqref{Boltzmann}. We therefore ask whether causal systems with purely relaxational spectra, but not expressed in a kinetic-like form, can evade our geometric constraints. The theorem below shows that, at least in the case of finite-order PDEs, this cannot happen. In fact, we have the following result:

\begin{theorem}\label{theoLax}
Consider a field $\phi:\textup{``Minkowski''}\to \mathbb{C}$, governed by a linear partial differential equation $\mathcal{P}(-\partial_t,\partial_x)\phi=0$, with $\mathcal{P}(y_1,y_2)$ a polynomial with real constant coefficients and finite degree $\mathfrak{n}>0$. Assume that all the rates $i\omega$ are real for real $ik$, and that the coefficient in front of $(-\partial_t)^{\mathfrak{n}}$ is $1$. Then, there exists a system of the form \eqref{Boltzmann} (with~$\mathcal{H}=\mathbb{C}^\mathfrak{n}$) whose dispersion relations and characteristic speeds are the same as those of $\phi$.
\end{theorem}

\begin{proof}
By setting $\phi=e^{ikx-i\omega t}$, we find that the dispersion relations of the system are the roots of the polynomial $\mathcal{P}(i\omega,ik)$. Consider the homogeneous polynomial $\mathcal{Q}(X,Y,Z)=Z^\mathfrak{n}\mathcal{P}(X/Z,Y/Z)$. This polynomial is hyperbolic relative to the direction $(1,0,0)$. In fact, $\mathcal{Q}(1,0,0)=\lim_{Z\to 0}Z^{\mathfrak{n}}\mathcal{P}(1/Z,0)=1$. Moreover, the univariate polynomial $\mathcal{R}(T)=\mathcal{Q}(X_0-T,Y_0,Z_0)=Z_0^\mathfrak{n}\mathcal{P}(X_0/Z_0-T/Z_0,Y_0/Z_0)$ has only real roots for any $(X_0,Y_0,Z_0)\in \mathbb{R}^3$, since these correspond to the values of $T$ for which the combination $X_0/Z_0-T/Z_0$ is a frequency $i\omega$ of the system, which must be real, since the corresponding wavenumber $ik=Y_0/Z_0$ is real (reality of $T$ in the case with $Z_0=0$ holds by continuity). Thus, we can apply the Lax Conjecture (which is proven \cite{LaxConjecture2003}), and conclude that there exist two Hermitian $\mathfrak{n}\times\mathfrak{n}$ matrices $\E$ and $\s$ such that $\mathcal{Q}(X,Y,Z)=\det(X-\E Y-\s Z)$. Setting $Z=1$, we obtain
\begin{equation}
\mathcal{P}(i\omega,ik)=\det(i\omega-ik\E-\s) \, .
\end{equation}
Hence, the roots of $\mathcal{P}(i\omega,ik)$ coincide with the values of $i\omega$ for which \eqref{boltzmannFourier} holds, for some $\Psi\in \mathbb{C}^\mathfrak{n}$. Therefore, the dynamical system $\partial_t\Psi=-\s \Psi-\E \partial_x \Psi$ constructed in this way possesses the same dispersion relations as $\phi$. To prove that the characteristics are also the same, we set $X=-\xi_0$, $Y=\xi_1$, and $Z=0$. The result is
\begin{equation}
\lim_{Z\to 0} Z^\mathfrak{n}\mathcal{P}(-\xi_0/Z,\xi_1/Z)=\det(-\xi_0-\xi_1 \E) \, .
\end{equation}
The limit on the left side extracts the characteristic polynomial of $\mathcal{P}(-\partial_t,\partial_x)\phi=0$ (i.e. the homogeneous part of $\mathcal{P}$ with degree $\mathfrak{n}$) \cite[\S 1.4]{rauch2012partial}. The determinant on the right side is the characteristic polynomial of $\partial_t\Psi=-\s \Psi-\E \partial_x \Psi$. Hence, the characteristics must also be the same.
\end{proof}

Let us briefly comment on the assumptions behind this theorem. The fact that we are considering a single scalar $\phi$ does not limit the generality of the result, since a linear system with constant coefficients $\sum_m a_{nm}(\partial_t,\partial_x)\phi_m=0$ can always be reduced to a single PDE $\det[a_{nm}(\partial_t,\partial_x)]\phi=0$ \cite[\S I.2.2]{CourantHilbert2_book} with the same spectrum. Moreover, the case where $\mathcal{P}$ has complex coefficients can be reduced to a system with real coefficients for the functions $\mathfrak{Re}\phi$ and $\mathfrak{Im}\phi$, which can in turn be reduced to a single equation for $\mathfrak{Re}\phi$ (and we can formally complexify $\mathfrak{Re}\phi$). The assumption that the coefficient in front of $(-\partial_t)^{\mathfrak{n}}$ is $1$ automatically holds if the dynamics is causal. In fact, if that coefficient is $0$, the characteristic polynomial $\mathcal{P}_{\mathfrak{n}}(-\xi_0,\xi_1)$ has a factor $\xi_1$, resulting in $t=\text{const}$ being characteristic lines, signaling infinite propagation speed. If that coefficient is a non-zero number $a$, we can always divide the PDE by $a$.

In conclusion, Theorem \ref{theoLax} tells us that any finite system of PDEs that is both \textbf{(a)} causal and \textbf{(b)} purely relaxational can be recast (for spectral purposes) in the form \eqref{Boltzmann}, where $w=||\E||$ coincides with the fastest characteristics speed of $\phi$ (and thus does not exceed $1$). As a consequence, the geometric structures and bounds derived in this work are not peculiar to self-adjoint kinetic or transient hydrodynamic theories, but are universal features of causal relaxational dynamics.

Let us discuss a couple of additional implications of Theorem \ref{theoLax}.

Consider a causal PDE of the form $(\partial_t{+}c_s\partial_x)\phi+``\text{higher derivatives''}{=}0$, or $(\partial_t^2{-}c_s^2\partial^2_x)\phi+``\text{higher derivatives''}{=}0$. These equations are models of sound, since they admit a hydrodynamic dispersion relation of the form $i\omega{=}c_sik+...$. Combining Theorems \ref{theoLax} and \ref{theo2}, we immediately conclude that, if the spectrum is purely relaxational, then $|c_s|\leq w$. Conversely, if $|c_s|>w$, then some excitation must become oscillatory (i.e. $i\omega\notin \mathbb{R}$) for some real $ik$, thereby leaving the relaxation plane.

Another remarkable aspect of Theorem \ref{theoLax} is that it clarifies the origin of the spectral similarity between BDNK and Israel--Stewart theory. Consider a choice of hydrodynamic frame such that the corresponding BDNK spectrum is purely relaxational. Then, there is a system of the form \eqref{Boltzmann} whose spectrum is identical to that of the given BDNK theory. However, when the number of degrees of freedom is small, all theories of the form \eqref{Boltzmann} are known to fall into universality classes, whose representative are various formulations of Israel--Stewart theory \cite{GavassinoUniveraalityI2023odx,GavassinoUniversalityII:2023qwl}. Hence, if a BDNK theory is purely relaxational, its spectrum is necessarily \textit{identical} to that of a corresponding Israel--Stewart theory.

\section{Conclusions}

In this work, we showed that relativistic theories with purely relaxational spectra naturally endow the $\{i\omega,ik\}$ plane with a Lorentzian geometric structure. Within this framework, the excitation spectrum of a medium behaves as a causal object propagating on the relaxation plane, with causality constraining dispersion relations to follow spacelike trajectories. This geometric perspective transforms a broad class of linear-response problems into elementary questions about the causal structure of the spectrum.

The resulting framework provides a unified language for phenomena that are usually discussed separately, including the breakdown of hydrodynamic expansions, the observer-dependence of relaxation processes, the structure of wakes and bow waves at large boosts, and the morphology of non-hydrodynamic sectors in kinetic theory. At the same time, it reveals that these apparently disparate phenomena are subject to a common set of geometric constraints. In this sense, the relaxation plane plays a role analogous to that of spacetime in relativistic kinematics: once the causal structure is known, many qualitative and quantitative properties become visible directly from geometry, often without solving the microscopic dynamics explicitly.

At a technical level, the geometric picture leads to rigorous and universal constraints on relativistic transport. Among other results, we derived sharp bounds on diffusivity and viscosity, established lower bounds on the regime of validity of hydrodynamics (including in boosted frames), quantified the maximal deviations from time dilation allowed by causality, and obtained general criteria governing the breakdown of dispersion relations, the observer dependence of spectral hierarchies, and the emergence of ballistic continua in kinetic theory. More broadly, the framework provides a systematic way of converting geometric information about the spectrum into quantitative physical constraints.

The present construction relies on the spectrum being purely relaxational. An important open question is the extent to which the resulting geometric picture persists in systems with genuinely oscillatory spectra, as occurs in strongly coupled holographic matter \cite{KovtunHolography2005,HellerBounds2022ejw}. More generally, it remains to be understood whether a broader framework exists that unifies oscillatory and relaxational dynamics within a common geometric language.

\section*{Acknowledgements}

I am grateful to M. Disconzi for reading the manuscript and providing useful feedback.
I also thank H. Reall and B. Withers for stimulating discussions. This work is supported by a MERAC Foundation prize grant,  an Isaac Newton Trust Grant, and funding from the Cambridge Centre for Theoretical Cosmology.

\appendix
\section{Wakes and bow waves}\label{aaa}

In this appendix, we briefly discuss the interpretation of evanescent-like modes as wakes and bow waves.

Consider a medium at rest, and an immersed object moving to the right with velocity $v>0$. We model the object as an external force of the form $F(t,x)\propto \delta(x-vt)$. Conventionally, the front of the object is taken to point in the direction of motion (for example, the bow of a ship moving through water) so that the bow-wave region corresponds to $x>vt$, while the wake occupies the region $x<vt$ \cite{NovakWithersObstacles:2018pnv}.

To model these structures, we boost to the rest frame of the object, $(t',x')=\gamma(t-vx,x-vt)$,
where the force becomes stationary, $F\propto\delta(x')$. In these coordinates, the equation of motion for a field $\phi$ describing the perturbation takes the form
\begin{equation}
\mathcal{L}(\partial_{t'},\partial_{x'}) \phi =F(x') \, ,
\end{equation}
where $\mathcal{L}(\partial_{t'},\partial_{x'})$ is a linear differential operator. Assuming stationarity in the comoving frame, we must solve the ordinary differential equation $\mathcal{L}(0,\partial_{x'}) \phi =F(x')$.
The general solution that remains finite at spatial infinity is
\begin{equation}\label{generalwakebow}
\phi=
\begin{cases}
\sum_n \mathcal{P}_n(x') e^{s_n x'} & \text{for }x'<0 \, , \\
\sum_m \mathcal{Q}_m(x') e^{-s_m x'} & \text{for }x'>0 \, ,
\end{cases}
\end{equation}
supplemented with appropriate matching conditions at $x'=0$. Here, the coefficients $s_n$ correspond to evanescent modes lying in the right Rindler wedge, while the $s_m$ correspond to modes in the left wedge. The functions $\mathcal{P}_n(x')$ and $\mathcal{Q}_m(x')$ are polynomials whose degree depends on the multiplicity of the corresponding mode.
We therefore conclude that, for $v>0$, modes in the right wedge describe wakes, while modes in the left wedge describe bow waves.

For $v<0$, the solution \eqref{generalwakebow} still applies, but the front of the object now points toward negative $x'$, so the distinction between wake and bow wave is reversed.

Below, we summarize these results in a mnemonic table:
\begin{center} \begin{tabular}{|c|c| c|} \hline & Left wedge & Right wedge \\ \hline $v<0$ & Wake & Bow \\ \hline $v>0$ & Bow & Wake \\ \hline \end{tabular} \end{center}

\section{Singularities of retarded correlators}\label{bbb}

In the context of classical PDEs, retarded correlators reduce to Green functions \cite{MullinsInfo2023tjg,GavassinoConsistentFluctuations:2024vyu}. In particular, one considers perturbing the equilibrium state with a localized source $A$ inserted at $(t,x)=(0,0)$, and studying the induced response in another observable $B$ at a different event. In kinetic-type theories governed by \eqref{Boltzmann}, both observables may be represented as elements of $\mathcal{H}$. Introducing the source term $A$ on the right-hand side of \eqref{Boltzmann},
\begin{equation}
(\partial_t+\s+\E\partial_x)\Psi
=
A\delta(t)\delta(x) \, ,
\end{equation}
the corresponding response function is the quantity $(B,\Psi)$. Passing to Fourier space, the associated Green function is
\begin{equation}\label{greenfunction}
\mathcal{G}_{AB}(i\omega,ik)
=
\left(B,(\s_{ik}-i\omega)^{-1}A\right) \, .
\end{equation}

We now prove the following result.

\begin{theorem}
Suppose that $\s$ and $\E$ are self-adjoint, with $\E$ bounded, and let $ik\in\mathbb{R}$. Then, $(i\omega,ik)$ belongs to the spectrum of $\s_{ik}$ if and only if, for every open neighborhood $\mathcal{U}$ of $(i\omega,ik)$, there exist $A,B\in\mathcal{H}$ such that $\mathcal{G}_{AB}$ fails to be holomorphic somewhere in $\mathcal{U}$.
\end{theorem}

\begin{proof}
The key observation is that $(\s_{ik}-s)^{-1}$ is the resolvent of $\s_{ik}$ evaluated at $s$, and therefore becomes singular precisely when $s$ belongs to the spectrum.

Suppose first that $i\omega$ does \textit{not} belong to the spectrum of $\s_{ik}$. The family $\s_{ik}$ is holomorphic in the operator sense, being the sum of the trivially holomorphic operator $\s$ and the bounded-holomorphic family $ik\E$ \cite[\S 7.1.2,~Problem~1.2]{Kato_Perturbation_Theory}. Moreover, each $\s_{ik}$ is closed. We may therefore apply a standard result from analytic perturbation theory \cite[\S 7.1.2, Theorem 1.3]{Kato_Perturbation_Theory}, which implies that the resolvent $(\s_{ik}-s)^{-1}$ is bounded-holomorphic in an open neighborhood $\mathcal{U}$ of $(i\omega,ik)$. Consequently, every matrix element \eqref{greenfunction} is well defined and holomorphic throughout $\mathcal{U}$.

Conversely, we prove the contrapositive. Fix $ik\in\mathbb{R}$, and assume that there exists an open neighborhood $\mathcal{U}$ of $(i\omega,ik)$ such that, for all $A,B\in\mathcal{H}$, the Green functions $\mathcal{G}_{AB}$ are holomorphic throughout $\mathcal{U}$. Taking the slice at fixed $ik\in  \mathbb{R}$, we obtain a neighborhood $\mathcal{V}$ of $i\omega$ such that, for all $A,B\in\mathcal{H}$, the functions $\left(B,(\s_{ik}-s)^{-1}A\right)$ are holomorphic for $s\in\mathcal{V}$.
Let $2\epsilon$ denote the radius of a disk centered at $i\omega$ and contained in $\mathcal{V}$. Since $\s_{ik}$ is self-adjoint, its resolvent exists away from the real axis and satisfies $||(\s_{ik}-s)^{-1}||\leq 1/|\mathfrak{Im}\,s|$. Hence, along the sequence $s_n=i\omega+i\epsilon(1+1/n)\in\mathcal{V}$, the sesquilinear forms $\left(*,(\s_{ik}-s_n)^{-1}*\right)$ are uniformly bounded. Combined with holomorphicity, the principle of uniform boundedness \cite[\S 7.4.1, Remark 4.1]{Kato_Perturbation_Theory} then implies that the sesquilinear form $\left(*,(\s_{ik}-s)^{-1}*\right)$ extends to a bounded-holomorphic family throughout $\mathcal{V}$, including at $i\omega$. It follows that $(\s_{ik}-s)^{-1}$ itself is bounded-holomorphic in $\mathcal{V}$, and therefore that $i\omega$ belongs to the resolvent set, namely the complement of the spectrum.
\end{proof}

\section{Equivalence between causality and the operator bound}\label{ccc}

Here, we prove that equation \eqref{Boltzmann} is causal if and only if $w\equiv ||\E||$ does not exceed $1$. When $\mathcal{H}$ is finite-dimensional, the proof is straightforward. Indeed, equation \eqref{Boltzmann} reduces to a symmetric hyperbolic system of PDEs whose characteristic speeds $w_n$ satisfy $\det(\E-w_n)=0$ \cite{Hishcock1983}. These are simply the eigenvalues of $\E$, and the requirement that $w_n\in[-1,1]$ for all $n$ is equivalent to the bound $||\E||\leq1$.

The infinite-dimensional case lies outside the scope of standard PDE theory, and different techniques are required. To avoid unnecessary complications, we will assume that the initial-value problem associated with \eqref{Boltzmann} is well-posed forward in time. This requires $\s$ to be bounded from below. Moreover, since shifting $\s\mapsto\tilde{\s}=\s-a$ with $a\in\mathbb{R}$ simply rescales every solution by an overall factor $e^{-at}$, leaving domains of dependence unchanged, we will assume without loss of generality that $\s$ is non-negative definite (since we can always shift the lower bound on the spectrum above zero). In practice, this is not a restrictive assumption, since descriptions of stable phases of matter are required to contain no unstable-like modes, which implies that $\s$ is non-negative definite (set $ik=0$ in \eqref{boltzmannFourier}).

\subsubsection{The operator bound implies causality}

Let us first show that, if $w\leq 1$, then the evolution is causal. 

We follow the same strategy as in \cite{GavassinoCausality2021}. In particular, taking the inner product of \eqref{Boltzmann} with $\Psi$, adding to the result its complex conjugate, and recalling that $\s$ and $\E$ are self-adjoint, we obtain the identity
\begin{equation}\label{infona}
\partial_t \dfrac{(\Psi,\Psi)}{2}+\partial_x \dfrac{(\Psi,\E\Psi)}{2}+(\Psi,\s\Psi)=0 .
\end{equation}
We can use this balance law to prove that, if $\Psi(t,x)$ vanishes for $x\geq 0$ at $t=0$, then it must also vanish for $x\geq wt$ at $t>0$, which is precisely the statement of causality (for $w\leq 1$). To this end, we consider the family of integrals
\begin{equation}
\mathcal{I}(t)=\int_{wt}^{\infty} \dfrac{(\Psi,\Psi)}{2}\bigg|_{\text{at }(t,x)} dx\, .
\end{equation}
Clearly, $\mathcal{I}(t)\geq 0$ and $\mathcal{I}(0)=0$. On the other hand, using \eqref{infona} and \cite[\S 2.3, Eq. (2.65)]{TeschlBook}, we find that
\begin{equation}
\begin{split}
\dfrac{d\mathcal{I}}{dt}={}& -w\dfrac{(\Psi,\Psi)}{2}\bigg|_{\text{at }(t,wt)}+\int_{wt}^{\infty}\partial_t\dfrac{(\Psi,\Psi)}{2}\bigg|_{\text{at }(t,x)} dx \\
={}& -w\dfrac{(\Psi,\Psi)}{2}\bigg|_{\text{at }(t,wt)} +\dfrac{(\Psi,\E\Psi)}{2}\bigg|_{\text{at }(t,wt)}-\int_{wt}^{\infty}(\Psi,\s\Psi)\big|_{\text{at }(t,x)} dx\leq 0\, , \\
\end{split}
\end{equation}
where we invoked the fact that $\s$ is non-negative definite and $w=||\E||$.
Thus, $\mathcal{I}(t)$ must vanish at all positive times, implying that $\Psi(t,x)=0$ for $x\geq wt>0$, as required.

\subsubsection{Causality implies the operator bound}

We now prove that, if $w>1$, then either the evolution is acausal, or the initial value problem becomes ill-posed in some boosted frame, in the sense that solutions fail to exist for some smooth square-integrable initial data. Either outcome is incompatible with the standard principles of special relativity.

Suppose that $w>1$. Then there exists a vector $\Phi\in\mathcal{H}$ with unit norm such that $|(\Phi,\E\Phi)|>1$ \cite[\S 2.3, Eq. (2.65)]{TeschlBook}. Without loss of generality, we may assume that $E\equiv(\Phi,\E\Phi)$ is positive (otherwise, one may reverse the orientation of the $x$ axis). Fix any $v\in(E^{-1},1)$, and consider the initial value problem with data prescribed on the spacelike line $t=vx$. Assuming that \eqref{Boltzmann} admits a well-posed initial value formulation in every inertial frame (and in particular in the frame moving with velocity $v$) we may construct solutions with arbitrary sufficiently regular data on this line. In particular, we may choose a solution satisfying $\Psi(vx,x)=f(x)\Phi$, where $f$ is smooth and compactly supported.
By causality, the support of this solution must remain confined within the future light cone of the support of $f$. One may therefore construct a triangle as in Fig.~\ref{fig:DiagramForCausality}. Integrating \eqref{infona} over this triangle yields
\begin{equation}
\int_{t=t_f} \dfrac{(\Psi,\Psi)}{2} dx -\int_{t=vx} \left[\dfrac{(\Psi,\Psi)}{2}-v\dfrac{(\Psi,\E\Psi)}{2}\right] dx+\int_{\text{triangle}} (\Psi,\s\Psi) dt dx =0 \, ,
\end{equation}
where causality was used to set $\Psi$ to zero on the left side of the triangle. Since $\s$ is non-negative definite and $\Psi=f\Phi$ on the hypotenuse, it follows that
\begin{equation}
(1-vE)\int \dfrac{|f|^2}{2} dx \geq 0 \, .
\end{equation}
However, by construction, $v>E^{-1}$, and therefore $1-vE<0$, giving a contradiction. We conclude that, if the evolution is causal and the initial value problem is well posed in every inertial frame, then such a vector $\Phi$ cannot exist, and therefore $w\leq1$.
\begin{figure}[h!]
    \centering
\includegraphics[width=0.4\linewidth]{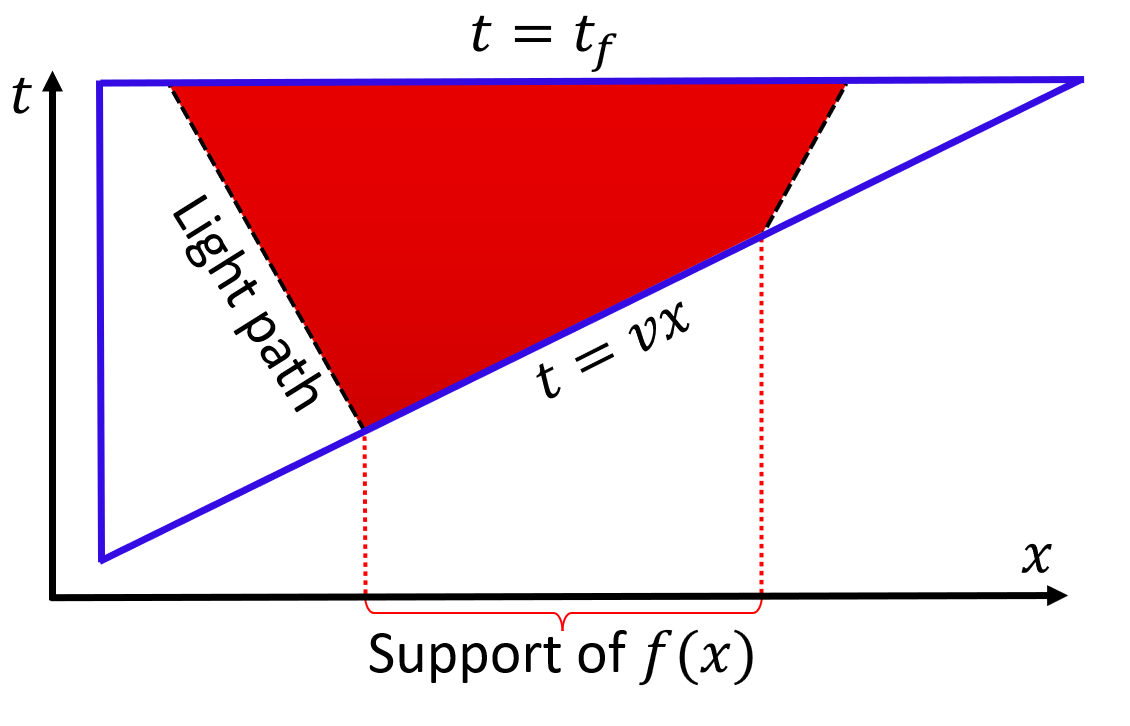}
\caption{Minkowski diagram illustrating that assuming causality, covariant well-posedness, and $||\E||>1$ leads to a contradiction. Indeed, well-posedness in boosted reference frames together with $||\E||>1$ implies that the system can be initialized on a spacelike line $t=vx$ with compactly supported data satisfying $\Psi\propto\Phi$, where $(\Phi,\E\Phi)>v^{-1}>1$ and $(\Phi,\Phi)=1$. By causality, the support of $\Psi$ can then expand at most at the speed of light. One may therefore construct a triangular spacetime region bounded by a constant-time slice $t=t_f$, the spacelike line $t=vx$, and a timelike boundary that does not intersect the support of $\Psi$. Integrating \eqref{infona}, which has the structure of a balance law $\partial_tE^0+\partial_xE^1=-\sigma\leq0$, over this region and applying Gauss' theorem implies that the flux through the final-time slice must be smaller than or equal to the flux through the line $t=vx$. However, the former is necessarily non-negative, since $E^0=(\Psi,\Psi)/2\geq0$, while the latter is strictly negative by construction, since $E^0-vE^1\propto 1-v(\Phi,\E\Phi)<0$.}
    \label{fig:DiagramForCausality}
\end{figure}

\bibliography{Biblio}

\label{lastpage}
\end{document}